\documentclass{article}

\usepackage{hyperref}
\usepackage{xcolor}
\usepackage{etoolbox}
\hypersetup{hidelinks}

\usepackage{amsthm}
\usepackage{mathtools}
\usepackage{bm}
\usepackage{newtxtext}
\usepackage[subscriptcorrection]{newtxmath}
\usepackage{bbm}

\usepackage{booktabs}

\usepackage[numbers,sort&compress]{natbib}

\usepackage{enumitem}

\usepackage{aliascnt}
\usepackage{cleveref}

\usepackage{needspace}

\newtheorem{theorem}{Theorem}[section]
\newaliascnt{definition}{theorem}
\newtheorem{definition}[definition]{Definition}
\aliascntresetthe{definition}
\newaliascnt{proposition}{theorem}
\newtheorem{proposition}[proposition]{Proposition}
\aliascntresetthe{proposition}
\newaliascnt{lemma}{theorem}
\newtheorem{lemma}[lemma]{Lemma}
\aliascntresetthe{lemma}
\newaliascnt{corollary}{theorem}
\newtheorem{corollary}[corollary]{Corollary}
\aliascntresetthe{corollary}
\newaliascnt{remark}{theorem}
\newtheorem{remark}[remark]{Remark}
\aliascntresetthe{remark}

\AddToHook{env/theorem/begin}{\crefalias{section}{theorem}}
\AddToHook{env/definition/begin}{\crefalias{section}{definition}}
\AddToHook{env/proposition/begin}{\crefalias{section}{proposition}}
\AddToHook{env/lemma/begin}{\crefalias{section}{lemma}}
\AddToHook{env/corollary/begin}{\crefalias{section}{corollary}}
\AddToHook{env/remark/begin}{\crefalias{section}{remark}}

\crefname{theorem}{Theorem}{Theorems}
\crefname{definition}{Definition}{Definitions}
\crefname{proposition}{Proposition}{Propositions}
\crefname{lemma}{Lemma}{Lemmas}
\crefname{corollary}{Corollary}{Corollaries}
\crefname{remark}{Remark}{Remarks}
\Crefname{theorem}{Theorem}{Theorems}
\Crefname{definition}{Definition}{Definitions}
\Crefname{proposition}{Proposition}{Propositions}
\Crefname{lemma}{Lemma}{Lemmas}
\Crefname{corollary}{Corollary}{Corollaries}
\Crefname{remark}{Remark}{Remarks}

\newcommand{\R}{\mathbb{R}}
\newcommand{\F}{\mathcal{F}}
\newcommand{\E}{\mathbb{E}}
\newcommand{\Prob}{\mathbb{P}}
\newcommand{\diff}{\mathrm{d}}
\newcommand{\norm}[1]{\left\lVert#1\right\rVert}

\title{Mean-field equilibrium of heterogeneous agents under market impact}
\author{Joseph Leclère\thanks{Ceremade, Université Paris Dauphine-PSL,
\texttt{leclere@ceremade.dauphine.fr}}
\and Mathieu Rosenbaum\thanks{Ceremade, Université Paris Dauphine-PSL,
\texttt{mathieu.rosenbaum@dauphine.psl.eu}}}

\begin{document}
\raggedbottom

\maketitle

\begin{abstract}
Although market participants generally have access to a common information
set, they make decisions based on forecasts formed over heterogeneous horizons.
Because market impact depends on aggregate positions rather than trader
identities, these decisions feed back into prices through their collective
effect. We introduce a linear mean-field model of this interaction. The
observed price is decomposed into a martingale component, a common predictable
signal represented by a Volterra process, and the market impact generated by
aggregate positions. Agents take positions according to conditional forecasts
of future signal increments and a fraction of anticipated aggregate impact over
their respective horizons. Within a Gaussian--Volterra framework, we
characterize equilibrium through a linear fixed-point equation for aggregate
positions and establish existence and uniqueness under explicit conditions.
At equilibrium, we identify a balance condition that cancels the direct
transmission of the common signal to the observed price. We then study the limit
in which agents fully account for market impact. Along a suitably scaled family of
equilibria satisfying explicit conditions, the contributions of the predictable signal and the market impact
cancel in the limit, and the observed price converges to its martingale
component. For fractional-type signals and Gamma-distributed horizons, we
further derive local Hölder bounds and identify the horizon distributions for
which %the full price has local regularity compatible with that of Brownian motion.
the observed price has local regularity compatible with that of Brownian motion.
\end{abstract}

\section{Introduction}\label{sec:introduction}

Forecasts used in financial markets are naturally heterogeneous in their
horizons. Market makers, execution desks, and directional traders do not act on
the same forecast window, even when their decisions draw on common information.
Their decisions aggregate into persistent signed order flow, which in turn moves
prices. This creates a consistency problem: forecasts made at different
horizons must aggregate through market impact in a way that reconciles
persistent order flow with weakly predictable price changes and a local
regularity close to that of a Brownian motion.\\

In our setting, agents respond to a common forecast signal, while a residual
martingale represents the part of observed price changes that is not
predictable. We study this consistency in two related senses: aggregate
trading should attenuate the predictable signal transmitted to prices, and the
resulting paths should recover a regularity compatible with that of a Brownian
motion.
Market impact links agents' positions to the observed price. We follow the
propagator approach to market impact
\cite{BouchaudGefenPottersWyart2004,Gatheral2010,Jaisson2014MarketImpact,
JusselinRosenbaum2020}, in which a kernel $G$ governs how the price effect induced by a position change evolves over time.
For an aggregate position $\Pi$,
initialized at $\Pi_0=0$, the resulting impact is written
\begin{equation*}
  \int_0^t G(t-s)\,\diff\Pi_s.
\end{equation*}
When the Stieltjes integral and integration by parts are justified,
\begin{equation*}
  \int_0^t G(t-s)\,\diff\Pi_s
  =
  G(0)\Pi_t
  +\int_0^t G'(t-s)\Pi_s\,\diff s.
\end{equation*}
This decomposition naturally separates the \emph{instantaneous impact}
$G(0)\Pi_t$ from the \emph{transient impact} $(G'\ast\Pi)_t$ generated by past
positions. The position-based expression on the right is the impact
specification used below.\\

Two distinct notions of market endogeneity are relevant. On the order-flow side,
Hawkes processes are now commonly used to model order arrivals, with statistical
self-excitation capturing their clustering and strong serial dependence
\cite{BouchaudFarmerLillo2009,BacryMastromatteoMuzy2015}. For a nonnegative
excitation kernel $h$, the
branching ratio $\lVert h\rVert_1$ is the expected number of direct offspring
generated by one event. Stability requires $\lVert h\rVert_1<1$, and values
close to one describe highly persistent, near-critical activity
\cite{JaissonRosenbaum2016,GatheralJaissonRosenbaum2018}. On the forecast side,
agents trade on forecasts of both the common signal and the collective impact
generated by all participants, with the latter creating the endogenous
feedback. The coefficient $\alpha \in [0,1]$ represents the fraction of
anticipated collective impact that enters their forecasts.\\

The model has three price components. A Brownian motion $B$ drives the common
forecastable signal
\begin{equation*}
  P_t=\int_0^t K(t-s)\,\diff B_s.
\end{equation*}
A prototypical fractional-type specification, studied specifically in
\Cref{sec:tail_regularity}, is $K(t)=t^{H-1/2}$ with $H<1/2$. It provides a
reduced-form description of rough predictability at the mid-frequency scale,
with the decay of $K$ representing the fading influence of past shocks.
Agents are indexed by their forecast horizon $\lambda>0$, and the probability
measure $\nu$ describes the distribution of these horizons across the agent
population. For a fixed $\alpha$, their individual positions
$\pi_t^{\lambda,\alpha}$ aggregate across horizons as
\begin{equation*}
  \Pi_t^\alpha
  =
  \int_{\R_+^*}\pi_t^{\lambda,\alpha}\,\nu(\diff\lambda).
\end{equation*}
This aggregate position generates the impact
\begin{equation*}
  I_t^\alpha
  =G(0)\Pi_t^\alpha
  +\int_0^tG'(t-s)\Pi_s^\alpha\,\diff s,
\end{equation*}
and the observed price is
\begin{equation*}
  S_t^\alpha=M_t+P_t+I_t^\alpha.
\end{equation*}
Here $M$ is a residual martingale carrying price changes that are observed when
they occur, but whose increments have zero conditional expectation.
Agents' positions with horizon $\lambda$ are determined by the forecast rule
\begin{equation*}
  \pi_t^{\lambda,\alpha}
  =
  \mathbb E\!\left[
    P_{t+\lambda}-P_t
    +\alpha\bigl(I_{t+\lambda}^\alpha-I_t^\alpha\bigr)
    \mid\F_t
  \right].
\end{equation*}
This specification of positions provides a simplified representation in which each agent's position is identified with the strength of their forecast.
The martingale term therefore disappears from the forecast equation without
disappearing from the observed price. At $\alpha=0$, positions use only the
common signal forecast, while at $\alpha=1$, they account for the whole anticipated
aggregate-impact increment. Values $\alpha<1$ allow for
incomplete attribution of collective impact, limited information, model
uncertainty, or different mandates. For tractability, we assume that all agents
share the same constant coefficient $\alpha$, and a natural extension would allow
it to depend on the horizon through $\alpha(\lambda)$.\\

The analysis proceeds in three steps.
\begin{itemize}
  \item We first characterize equilibria within the class of
  Gaussian--Volterra processes. The equilibrium condition is a fixed-point problem for the aggregate position, and we seek solutions represented by a Volterra kernel $\phi_\alpha$ as
  \begin{equation*}
    \Pi_t^\alpha
    =
    \int_0^t\phi_\alpha(t-s)\,\diff B_s.
  \end{equation*}
  Under this representation, imposing population aggregation and the forecast rule for $\nu$-almost every $\lambda$ yields a single linear equation for $\phi_\alpha$. In general, this equation
  can have multiple solutions. Under explicit conditions, however, it has a
  unique solution in a natural Banach space of decaying functions.

  \item We then study the signal-driven component of the equilibrium price,
  $X^\alpha:=P+I^\alpha$, by decomposing its Volterra kernel as
  \begin{equation*}
    \mathcal K_\alpha
    =
    \frac{1-(1-\alpha)G(0)}{1+\alpha G(0)}K
    +\mathcal R_\alpha,
  \end{equation*}
  where $\mathcal R_\alpha$ collects terms obtained from $K$ and $\phi_\alpha$
  through horizon averaging and impact convolution. Under the Gamma-regular
  hypotheses of \Cref{sec:tail_regularity}, these transformed contributions are
  less singular near the origin than $K$.
  The local regularity of the Volterra signal $P$ is governed by the behaviour of $K$ near
  the origin, as illustrated by the rough kernel $K(t)=t^{H-1/2}$. For this rough kernel, the terms in $\mathcal{R}_\alpha$ obtained by averaging shifted values
  of $K$ and $\phi_\alpha$ across the horizon distribution are less singular at the origin, while the impact-convolution contributions
  are smoothed by integration over time.
  The first term above, by contrast, transmits it directly. Setting its coefficient to zero
  gives the rough-signal cancellation condition
  \begin{equation*}
    G(0)(1-\alpha)=1,
  \end{equation*}
  which removes the direct rough-signal term. In the rough-signal setting, cancelling this direct $K$-term gives the signal-driven component $P + I^\alpha$ greater local Hölder regularity than $P$. 
  As $\alpha$ increases, agents incorporate a larger fraction of anticipated
  collective impact into their forecasts. In
  \Cref{subsec:hawkes_fixed_parameter}, we show that this degree of feedback is
  closely related to a standard measure of endogeneity in Hawkes processes.
  Since Hawkes processes are often studied close to their critical boundary, this link provides another motivation for studying the boundary
  $\alpha\uparrow1$.

  \item In \Cref{sec:scaling_limit}, we study the boundary regime
  $\alpha\uparrow1$ along the rough-signal cancellation condition, using a fixed
  normalized impact kernel. Under suitable conditions, the rescaled aggregate-position kernels converge and the signal-driven
  price component vanishes, so the observed price tends to the residual
  martingale $M$. For a fractional-type signal and
  a Gamma distribution of agents' horizons, we also show in \Cref{sec:tail_regularity} how the concentration of the horizon distribution $\nu$ near zero
  determines the regularity gained after the direct rough term has disappeared.
\end{itemize}

The rest of the paper is organized as follows. \Cref{sec:model_setup} defines
price formation, heterogeneous forecast horizons, and the equilibrium notion.
\Cref{sec:meanfield} derives the Gaussian--Volterra characterization, gives
conditions ensuring existence and uniqueness of the equilibrium kernel within
a natural class of decaying functions, and isolates the direct signal
coefficient.
\Cref{sec:scaling_limit} studies the limit of a scaled family of equilibria as
$\alpha\uparrow1$ and agents approach full impact accounting.
\Cref{sec:tail_regularity} specializes the model to rough signals and
Gamma horizons. The appendices contain the existence and uniqueness, scaling,
and local regularity arguments.

\section{Price formation and heterogeneous forecast horizons}
\label{sec:model_setup}

\subsection{Information available and price components}

Let $(\Omega,\F,(\F_t)_{t\geq0},\Prob)$ be a filtered probability space
carrying a one-dimensional $(\F_t)$-Brownian motion $B$ and a real-valued
continuous square-integrable $(\F_t)$-martingale $M$ with locally
$\rho$-H\"older paths for every $\rho<1/2$. The filtration
represents the information available to agents: it contains the forecasting
features summarized by the one-dimensional process $B$, aggregate positions and prices, and the realized
path of $M$.

The observed price has three components. The common forecastable component is
the Volterra process
\begin{equation*}
    P_t:=\int_0^t K(t-s)\,\diff B_s,
    \qquad t\geq0.
\end{equation*}
For each $\alpha\in[0,1]$, the signal-driven component and the observed price are, respectively,
\begin{equation*}
    X_t^\alpha:= P_{t} + I_t^\alpha,
    \qquad
    S_t^\alpha:=M_t+X_t^\alpha
    =M_t+P_t+I_t^\alpha.
\end{equation*}
Here $I^\alpha$ denotes the price impact generated by agents' collective
trading activity.
The component $P$ translates the forecasting features into price units and is
the common signal on which agents act, rather than a separately traded
fundamental price. The process $M$ represents the remaining price changes.
They are observed ex post, since $M=S^\alpha-P-I^\alpha$, but
\begin{equation*}
    \E[M_{t+\lambda}-M_t\mid\F_t]=0,
    \qquad t\geq0,\ \lambda>0,
\end{equation*}
so they do not contribute to directional forecasts.

\subsection{Heterogeneous horizons and impact accounting}

Agents form a continuum indexed by their forecast and decision horizon
$\lambda\in\R_+^*$. The probability measure $\nu$ on $\R_+^*$ is the
population distribution of these horizons. For a given $\alpha$, agent
$\lambda$ holds a progressively measurable position
$(\pi_t^{\lambda,\alpha})_{t\geq0}$. Whenever the family is jointly measurable
and integrable with respect to $\nu$, the aggregate position is defined by
\begin{equation*}
    \Pi_t^\alpha
    :=\int_{\R_+^*}\pi_t^{\lambda,\alpha}\,\nu(\diff\lambda).
\end{equation*}
Positions are initialized at zero, hence $\Pi_0^\alpha=0$.
Each agent is atomistic, and impact is generated collectively by
$\Pi^\alpha$. For an impact kernel $G$ satisfying the assumptions below, we
take the position-based impact equation
\begin{equation*}
    I_t^\alpha
    :=G(0)\Pi_t^\alpha
      +\int_0^t G'(t-s)\Pi_s^\alpha\,\diff s.
\end{equation*}
This form remains meaningful for progressively measurable positions that are
locally integrable, and we use it throughout the paper.\\

We identify positions with forecasts, thereby normalizing the common
forecast-to-position scale to one; any other common scale can be absorbed into
the impact kernel. Heterogeneity then enters only through the horizons.
Accordingly, for a given $\alpha$, an equilibrium is a family of positions satisfying, for $\nu$-almost
every $\lambda>0$,
\begin{equation}
\label{eq:mean_field_eq}
\begin{aligned}
    \pi_t^{\lambda,\alpha}
    &=\E\!\left[
        P_{t+\lambda}-P_t
        +\alpha\bigl(I_{t+\lambda}^\alpha-I_t^\alpha\bigr)
        \,\middle|\,\F_t
    \right],
    \qquad t\geq0.
\end{aligned}
\end{equation}
At $\alpha=1$, the martingale property of $M$ gives
\begin{equation*}
    \pi_t^{\lambda,1}
    =\E[S_{t+\lambda}^1-S_t^1\mid\F_t].
\end{equation*}
Thus full impact accounting means forecasting the future price increment but
does not by itself eliminate the common component $P$.

\subsection{Assumptions and notations}

Throughout the paper, $G$ and $K$ are assumed to satisfy
\begin{equation*}
    G\in\mathcal C^2([0,\infty)),
    \qquad G'\in L^1(\R_+),
    \qquad G'(t)\longrightarrow0
    \quad\text{as }t\to\infty,
    \qquad G(0)\geq0,
\end{equation*}
and
\begin{equation*}
    K\in L^2_{\mathrm{loc}}(\R_+),
    \qquad K(t)\longrightarrow0
    \quad\text{as }t\to\infty.
\end{equation*}
The condition on $K$ means that the loading of an old signal shock eventually fades.
For a measurable function $f$ for which the expressions below are defined, we
use the following notation for $\tau>0$:
\begin{align*}
    \Delta_\lambda f(\tau)
    &:=f(\tau+\lambda)-f(\tau),
    \\
    \mathcal S_\nu f(\tau)
    &:=\int_{(0,\infty)}f(\tau+\lambda)\,\nu(\diff\lambda),
    \\
    \mathcal D_\nu f(\tau)
    &:=\int_{(0,\infty)}\Delta_\lambda f(\tau)\,\nu(\diff\lambda)
      =\mathcal S_\nu f(\tau)-f(\tau),
    \\
    (g\ast f)(\tau)
    &:=\int_0^\tau g(\tau-s)f(s)\,\diff s,
    \\
    \mathcal Hf
    &:=G'\ast f,
    \\
    \mathcal If
    &:=G(0)f+\mathcal Hf.
\end{align*}
The operator $\mathcal S_\nu$ averages a function across forecast horizons,
$\mathcal D_\nu$ averages its horizon increments, and $\mathcal I$ maps an
aggregate position to its impact. Moreover, for any operator $\mathcal T$, we write $\mathcal T^0 := \mathrm{Id}$ and, for $n \geq 1$, write $\mathcal T^n$ for its $n$-fold iterate whenever the composition is well-defined.
For a bounded linear operator $\mathcal T:\mathcal X\to\mathcal Y$, we
denote its operator norm by
$\norm{\mathcal T}_{\mathcal X\to\mathcal Y}$.\\

Regarding path regularity, we say that a stochastic process has critical H\"older
exponent $\gamma>0$ if, almost surely, its paths are locally $\rho$-H\"older
for every $\rho \in (0,\gamma)$ and are not locally $\rho$-H\"older for any
$\rho>\gamma$. \\

Finally, a function $f$ is called
$\nu$-admissible if, for every $T>0$,
\begin{equation*}
    \int_{(0,\infty)}
        \norm{\Delta_\lambda f}_{L^2(0,T)}
    \,\nu(\diff\lambda)<\infty.
\end{equation*}

\section{Gaussian--Volterra equilibria and price consistency}
\label{sec:meanfield}

\subsection{Characterization of Gaussian--Volterra equilibria}
\label{subsec:derivation}

Throughout this subsection, $\alpha\in[0,1]$ is fixed.
Equation~\eqref{eq:mean_field_eq} is a fixed-point
condition for the aggregate position. Starting Picard iteration from the zero
position preserves the class of Gaussian--Volterra processes, which motivates
the representation
\begin{equation*}
\Pi_t^\alpha=\int_0^t\phi_\alpha(t-s)\,\diff B_s,
\qquad
\phi_\alpha\in L^2_{\mathrm{loc}}(\mathbb R_+).
\end{equation*}
The next result gives a necessary and sufficient condition within this class.
We require $\nu$-admissibility as a technical condition allowing stochastic
Fubini to be applied when aggregating positions across horizons.

\begin{proposition}[Equilibrium kernel equation]
\label{prop:meanfield_eq}
Let $\phi_\alpha\in L^2_{\mathrm{loc}}(\mathbb R_+)$, and suppose that $K$,
$\phi_\alpha$, and $\mathcal I\phi_\alpha$ are $\nu$-admissible. Then
\begin{equation*}
\Pi_t^\alpha=\int_0^t\phi_\alpha(t-s)\,\diff B_s
\end{equation*}
is the aggregate position of an equilibrium if and only if
\begin{equation}
\label{eq:meanfield_alpha}
(\mathrm{Id}-\alpha\mathcal D_\nu\mathcal I)\phi_\alpha
=\mathcal D_\nu K.
\end{equation}
In that case, for $\nu$-almost every $\lambda>0$, the individual
positions are
\begin{equation*}
\pi_t^{\lambda,\alpha}
=
\bigl(\Delta_\lambda(K+\alpha\mathcal I\phi_\alpha)\ast \diff B\bigr)_t,
\end{equation*}
and their $\nu$-aggregate is $\Pi^\alpha$.
\end{proposition}

\begin{proof}
For a Gaussian--Volterra aggregate position,
\begin{equation*}
I_t^\alpha
=G(0)\Pi_t^\alpha+(G'\ast\Pi^\alpha)_t
=\bigl(\mathcal I\phi_\alpha\ast \diff B\bigr)_t.
\end{equation*}
The horizon-$\lambda$ position is
\begin{align*}
\pi_t^{\lambda,\alpha}
&=
\E\!\left[
P_{t+\lambda}-P_t
+\alpha(I_{t+\lambda}^\alpha-I_t^\alpha)
\mid\F_t
\right]
\\
&=
\bigl(\Delta_\lambda(K+\alpha\mathcal I\phi_\alpha)\ast \diff B\bigr)_t.
\end{align*}
By $\nu$-admissibility, stochastic Fubini yields
\begin{equation*}
\int_{(0,\infty)}\pi_t^{\lambda,\alpha}\,\nu(\diff\lambda)
=
\bigl(\mathcal D_\nu(K+\alpha\mathcal I\phi_\alpha)\ast \diff B\bigr)_t.
\end{equation*}
Identifying Volterra kernels in $L^2_{\mathrm{loc}}$ by It\^o isometry gives
\begin{equation*}
\phi_\alpha
=\mathcal D_\nu K+\alpha\mathcal D_\nu\mathcal I\phi_\alpha,
\end{equation*}
which is equivalent to \eqref{eq:meanfield_alpha}. Conversely, if
\eqref{eq:meanfield_alpha} holds, the same calculation in reverse shows that
the displayed individual positions satisfy \eqref{eq:mean_field_eq} for almost
every horizon with respect to $\nu$ and aggregate to $\Pi^\alpha$.
\end{proof}

Using
\begin{equation*}
\mathcal D_\nu=\mathcal S_\nu-\mathrm{Id},
\qquad
\mathcal I=G(0)\mathrm{Id}+\mathcal H,
\end{equation*}
the kernel equation becomes
\begin{equation*}
\bigl((1+\alpha G(0))\mathrm{Id}
-\alpha G(0)\mathcal S_\nu
-\alpha\mathcal D_\nu\mathcal H\bigr)\phi_\alpha
=\mathcal D_\nu K.
\end{equation*}
With
\begin{equation*}
\psi_\alpha:=(1+\alpha G(0))\phi_\alpha,
\qquad
c_\alpha:=\frac{\alpha G(0)}{1+\alpha G(0)},
\qquad
d_\alpha:=\frac{\alpha}{1+\alpha G(0)},
\end{equation*}
this is equivalently
\begin{equation}
\label{eq:coboundary_fixed_point}
(\mathrm{Id}-c_\alpha\mathcal S_\nu
-d_\alpha\mathcal D_\nu\mathcal H)\psi_\alpha
=\mathcal D_\nu K.
\end{equation}
This rewriting isolates the two transformations that later determine
regularity. The operator $\mathcal S_\nu$ shifts a potential singularity away from the
origin and averages it across horizons, whereas $\mathcal H$ integrates against
the transient-impact kernel. Their precise regularizing effects depend on the
local assumptions introduced in \Cref{sec:tail_regularity}.

\subsection{Existence and uniqueness}

The preceding characterization does not by itself guarantee existence.
Within the admissible class of \cref{prop:meanfield_eq}, suppose that
\eqref{eq:meanfield_alpha} admits a solution $\phi_\alpha^K$.
Then every solution in the same class is of the form
\begin{equation*}
\phi=\phi_\alpha^K+u,
\qquad
u\in\ker(\mathrm{Id}-\alpha\mathcal D_\nu\mathcal I).
\end{equation*}
The null space depends only on $(\alpha,G,\nu)$, whereas
$\phi_\alpha^K$ is driven by the common signal.
\Cref{rem:noninjective_feedback_mode} gives an explicit example in which
this null space is nontrivial and uniqueness therefore fails.

\begin{remark}[An exponential forecast--impact resonance]
\label{rem:noninjective_feedback_mode}
For $\alpha>0$, the deterministic horizon distribution
$\nu=\delta_\ell$, with $\ell>0$, and the constant impact $G(t)=c>0$ yield
$\mathcal H=0$, $\mathcal I=c\,\mathrm{Id}$, and
\begin{equation*}
\mathcal D_\nu f(\tau)=f(\tau+\ell)-f(\tau).
\end{equation*}
The homogeneous kernel equation then becomes
\begin{equation*}
    f=\alpha c\,\mathcal D_\nu f,
    \qquad
f(\tau+\ell)=\left(1+\frac{1}{\alpha c}\right)f(\tau).
\end{equation*}
Hence
\begin{equation*}
f(\tau)=e^{\kappa\tau},
\qquad
\kappa=\frac{1}{\ell}\log\left(1+\frac{1}{\alpha c}\right),
\end{equation*}
is a nonzero solution. 

The resulting aggregate-position process
$U_t=\int_0^t e^{\kappa(t-s)}\,\diff B_s$ exhibits a forecast--impact
resonance: when the exogenous signal kernel is zero, agents form positions
solely from anticipated impact. A nonzero aggregate position can therefore sustain
itself without any external signal. The function $f$ belongs to
$L^2_{\mathrm{loc}}$, and both $f$ and $\mathcal If=cf$ are $\nu$-admissible.

\end{remark}

The feedback mode responsible for nonuniqueness grows exponentially, and this
is precisely the behaviour we wish to exclude. In the next subsection, we
construct a unique solution in a Banach space of functions that decay at infinity. The
argument has two steps. When impact is only instantaneous, as reflected by
$G'=0$, the reduced equation is solved directly by an Abel pseudoinverse in
\Cref{subsubsec:no_transient_impact}. Reintroducing the transient-impact kernel
$G'$, we write the solution as the sum of the instantaneous-impact solution and
a residual term, and solve the resulting equation for that residual in
\Cref{subsubsec:transient_impact_perturbation}. Appendix~\ref{app:coboundary_existence}
proves that the resulting kernel is the unique solution of the kernel equation within this Banach space.

\subsubsection{Pseudoinverse without transient impact}
\label{subsubsec:no_transient_impact}

Assume first that $G'=0$, so that $\mathcal H=0$ and impact is purely
instantaneous. Equation~\eqref{eq:coboundary_fixed_point} reduces to
\begin{equation}
\label{eq:no_transient_coboundary}
(\mathrm{Id}-c_\alpha\mathcal S_\nu)\psi=\mathcal D_\nu K.
\end{equation}

Let $\mathcal B_0$ be the space of bounded Borel functions vanishing at
infinity, equipped with the supremum norm
\begin{equation*}
\mathcal B_0
:=
\{h:\mathbb R_+^*\to\mathbb R:\ h\text{ is bounded and Borel},\
h(t)\to0\text{ as }t\to\infty\},
\qquad
\norm{h}_{\mathcal B_0}:=\sup_{t>0}|h(t)|.
\end{equation*}
Suppose that $\mathcal S_\nu^nK\in\mathcal B_0$ for some $n\geq0$, and let
$n_0$ be the smallest such integer. We work on the Banach extension of
$\mathcal B_0$ defined by
\begin{equation*}
\mathcal Y_K^0
:=
\mathcal B_0
\oplus\mathrm{span}\{\mathcal S_\nu^jK:0\le j<n_0\}.
\end{equation*}
On $\mathcal Y_K^0$, the Abel operator
\begin{equation*}
Q_{c_\alpha}h
:=
\sum_{n\ge0}c_\alpha^n\mathcal S_\nu^n\mathcal D_\nu h
\end{equation*}
is a well-defined bounded operator on $\mathcal Y_K^0$ and satisfies
\begin{equation*}
(\mathrm{Id}-c_\alpha\mathcal S_\nu)Q_{c_\alpha}=\mathcal D_\nu.
\end{equation*}
These assertions are proved in \cref{lem:app_Qc_extension}. Hence
$Q_{c_\alpha}K$ solves \eqref{eq:no_transient_coboundary}.

\subsubsection{Transient-impact perturbation}
\label{subsubsec:transient_impact_perturbation}

We now restore $G'$ and seek the rescaled kernel as
\begin{equation*}
\psi_\alpha=Q_{c_\alpha}K+\eta_\alpha,
\qquad
\eta_\alpha\in\mathcal B_0.
\end{equation*}
The first term solves the instantaneous-impact equation. Substitution into
\eqref{eq:coboundary_fixed_point} shows that the transient-impact correction
must solve
\begin{equation*}
\eta_\alpha
=d_\alpha Q_{c_\alpha}\mathcal H
\bigl(Q_{c_\alpha}K+\eta_\alpha\bigr).
\end{equation*}
Thus transient impact enters as a fixed-point perturbation of the preceding
solution.
The fixed-point map is affine on $\mathcal B_0$. When it is a strict
contraction, its fixed point admits a convergent Neumann-series representation,
yielding the following equilibrium construction.

\needspace{13\baselineskip}
\begin{proposition}[Unique equilibrium kernel in $\mathcal Y_K^0$]
\label{prop:selected_coboundary_equilibrium}
Let $n_0$ be the smallest nonnegative integer such that
$\mathcal S_\nu^{n_0}K\in\mathcal B_0$. Assume that
\begin{equation*}
|d_\alpha|\,
\norm{Q_{c_\alpha}\mathcal H}_{\mathcal B_0\to\mathcal B_0}<1.
\end{equation*}
Then the kernel
\begin{equation}
\label{eq:selected_equilibrium_expansion}
\phi_\alpha
=
\frac{1}{1+\alpha G(0)}
\left[
Q_{c_\alpha}K
+\sum_{m\ge0}d_\alpha^{m+1}
(Q_{c_\alpha}\mathcal H)^m
Q_{c_\alpha}\mathcal H Q_{c_\alpha}K
\right]
\end{equation}
is the unique solution in $\mathcal Y_K^0$ of the kernel equation
\eqref{eq:meanfield_alpha}. Moreover, the process
\begin{equation*}
    \Pi_t^\alpha=\int_0^t\phi_\alpha(t-s)\,\diff B_s
\end{equation*}
is the aggregate position of a Gaussian--Volterra equilibrium.
\end{proposition}
The proof is given in \cref{subsec:app_transient_impact_perturbation}.

\begin{remark}[Contraction and permanent impact]
\label{rem:contraction_permanent_impact}
The operator norm appearing in \cref{prop:selected_coboundary_equilibrium} is
difficult to evaluate in general. Suppose that $G$ is nonnegative, decreasing,
and convex, with $G(0)>0$. Set
\begin{equation*}
\Theta
:=\frac{G(0)-G(\infty)}{G(0)}.
\end{equation*}
Thus $G(0)$ is the instantaneous impact level, $G(\infty)$ is the permanent
impact level, and $\Theta$ is the transient fraction. Since
$\norm{G'}_1=G(0)-G(\infty)$ and $d_\alpha=c_\alpha/G(0)$,
\begin{equation*}
|d_\alpha|\,
\norm{Q_{c_\alpha}\mathcal H}_{\mathcal B_0\to\mathcal B_0}
\leq2c_\alpha\Theta.
\end{equation*}
Consequently, the contraction condition holds whenever
$c_\alpha\Theta<1/2$. In particular, it holds whenever at least half of the
instantaneous impact is permanent.
\end{remark}

\subsection{Price decomposition and rough-signal cancellation}
\label{subsec:amplitude_factor}

Every equilibrium process determines an aggregate position and, through
collective market impact, a signal-driven component. Together with the
residual martingale $M$, this yields a corresponding observed price. The
price-level consistency question is then twofold: whether equilibrium trading
attenuates the transmission of the forecastable common signal, and whether the
resulting price has local regularity compatible with that of a Brownian motion.\\

Let $\phi_\alpha$ be the kernel of an equilibrium process. Since
$X^\alpha=P+I^\alpha$, stochastic Fubini gives
\begin{equation}
\label{eq:Ktilde_volterra}
S_t^\alpha
=M_t+X_t^\alpha
=M_t+\int_0^t\mathcal K_\alpha(t-s)\,\diff B_s,
\end{equation}
where
\begin{equation}
\label{eq:Ktilde_def}
\mathcal K_\alpha
:=K+\mathcal I\phi_\alpha
=K+G(0)\phi_\alpha+\mathcal H\phi_\alpha
\end{equation}
is the Volterra kernel of $X^\alpha$.
This decomposition separates the two sources of price variation. The
martingale residual $M$ has zero conditional expected increments and the
appropriate regularity, whereas $\mathcal K_\alpha$ governs how the
common signal is transmitted through equilibrium impact and the regularity of
the resulting signal-driven component.
To identify the direct contribution of the signal kernel $K$, which can
potentially induce roughness in the price, we rewrite $\mathcal K_\alpha$ as
follows.

\begin{proposition}[Direct-signal coefficient]
\label{prop:amplitude_factor}
The kernel of $X^\alpha$ satisfies
\begin{equation}
\label{eq:amplitude_identity}
\mathcal K_\alpha
=A(\alpha,G(0))K+\mathcal R_\alpha,
\end{equation}
where
\begin{equation*}
A(\alpha,x)
:=\frac{1-(1-\alpha)x}{1+\alpha x}
=1-\frac{x}{1+\alpha x},
\end{equation*}
and
\begin{equation}
\label{eq:amplitude_remainder}
\mathcal R_\alpha
=\mathcal H\phi_\alpha
+\frac{G(0)}{1+\alpha G(0)}
\left[
\mathcal S_\nu K
+\alpha G(0)\mathcal S_\nu\phi_\alpha
+\alpha(\mathcal S_\nu-\mathrm{Id})\mathcal H\phi_\alpha
\right].
\end{equation}
\end{proposition}

\begin{proof}
Substituting $\mathcal D_\nu=\mathcal S_\nu-\mathrm{Id}$ and
$\mathcal I=G(0)\mathrm{Id}+\mathcal H$ into
\eqref{eq:meanfield_alpha} gives
\begin{equation*}
(1+\alpha G(0))\phi_\alpha
=-K+\mathcal S_\nu K
+\alpha G(0)\mathcal S_\nu\phi_\alpha
+\alpha(\mathcal S_\nu-\mathrm{Id})\mathcal H\phi_\alpha.
\end{equation*}
Substitution of this identity into
$\mathcal K_\alpha=K+G(0)\phi_\alpha+\mathcal H\phi_\alpha$ and collection of
the coefficient of the untransformed kernel $K$ yield
\eqref{eq:amplitude_identity} and \eqref{eq:amplitude_remainder}.
\end{proof}

The direct contribution of $K$ vanishes precisely when
\begin{equation}
\label{eq:rough_signal_cancellation}
G(0)(1-\alpha)=1.
\end{equation}
We call \eqref{eq:rough_signal_cancellation} the \emph{rough-signal cancellation
condition}. For $\alpha<1$, it is equivalently
$G(0)=1/(1-\alpha)$, while no finite $G(0)$ satisfies it at $\alpha=1$.\\

Under the Gamma-regular hypotheses of \Cref{sec:tail_regularity}, we show that,
for a signal with regularity $H<1/2$, absence of cancellation forces $X^\alpha$ to retain the local
regularity of the signal, whereas after cancellation the remaining kernel has
higher local order and can reach the regularity of a Brownian motion. As such,
cancellation is necessary for the observed price to have local regularity
different from that of the signal.
\Cref{sec:scaling_limit} studies the regime $\alpha\uparrow1$ along the
cancellation condition and, under suitable conditions, proves that
$X^\alpha\to0$ and $S^\alpha\to M$. In this
limit, the signal-driven component disappears and only the residual martingale
remains. This limiting behaviour indicates that the cancellation condition also
reduces the predictability transmitted to prices by the common signal.

\subsection{Flow and forecast endogeneity}
\label{subsec:hawkes_fixed_parameter}

To model order flow and relate it to price impact, a standard approach combines
Hawkes processes with propagator models
\cite{BouchaudGefenPottersWyart2004,Jaisson2014MarketImpact}. Let $N^a$ and
$N^b$ be independent Hawkes processes counting, respectively, unit-sized buy
and sell market orders, with intensities
\begin{equation*}
\lambda_t^i
=\mu+\int_{[0,t)}h(t-s)\,\diff N_s^i,
\qquad i\in\{a,b\},\quad \mu>0,
\end{equation*}
where $h$ is nonnegative and
$\norm{h}_1:=\int_0^\infty h(s)\,\diff s<1$. The resulting order flow is
self-exciting and clustered. In its branching representation, each event has
$\norm{h}_1$ direct offspring on average
\cite{BouchaudFarmerLillo2009,BacryMastromatteoMuzy2015}. Absence of price
manipulation requires permanent impact to be linear in signed traded volume.
Combined with price efficiency, this yields
\begin{equation*}
\begin{aligned}
\widehat S_t
&=\widehat S_0
+\kappa\lim_{T\to\infty}
\mathbb E\!\left[N_T^a-N_T^b\mid\F_t\right] \\
&=\widehat S_0
+\kappa\int_0^t\Xi(t-s)\bigl(\diff N_s^a-\diff N_s^b\bigr),
\end{aligned}
\end{equation*}
as in \cite{Gatheral2010,Jaisson2014MarketImpact,JusselinRosenbaum2020}.
Here $\kappa>0$ fixes the price-impact scale of one order, and the associated
dimensionless propagator is
\begin{equation*}
\Xi(x)
:=1+\frac{1}{1-\norm{h}_1}\int_x^\infty h(s)\,\diff s,
\qquad
\Xi(\infty)=1,
\qquad
\Xi(0)=\frac{1}{1-\norm{h}_1}.
\end{equation*}
Our model measures positions in forecast units, whereas the Hawkes model
measures flow in unit trades, so their overall impact scales are not directly
comparable. We then compare the two dimensionless propagator kernels by
identifying $G=\Xi$. For fixed $\alpha<1$, the rough-signal cancellation
condition then becomes
\begin{equation*}
\norm{h}_1=\alpha.
\end{equation*}

Thus the identity $\norm{h}_1=\alpha$ links two a priori distinct notions of
endogeneity. On the forecast side, $\alpha$ quantifies the anticipation of
aggregate impact that agents incorporate into their positions. On the
order-flow side, $\norm{h}_1$ quantifies the self-excitation, and hence the
strong autocorrelation, underlying the Hawkes description. Equality matches
the strengths of these two mechanisms in the cancellation condition.

\section{Impact accounting at the boundary}
\label{sec:scaling_limit}

\subsection{Scaling regime and functional setting}

\subsubsection{Scaling regime}

Fix a kernel $G$ with $G(0)=1$ and define the scaled family
\footnote{This differs from the standard nearly unstable Hawkes scaling
$h_\alpha=\alpha h$, with $\norm{h}_1=1$
\cite{Jaisson2014MarketImpact,JaissonRosenbaum2016}. Under the identification
of \Cref{subsec:hawkes_fixed_parameter}, our parametrization
keeps $G_\alpha(\infty)/G_\alpha(0)=G(\infty)$ fixed and positive under the
conditions below. The propagator associated with $h_\alpha$ instead has
normalized permanent fraction $1-\alpha\to0$.}
\begin{equation*}
G_\alpha(t):=\frac{G(t)}{1-\alpha},
\qquad 0\leq\alpha<1.
\end{equation*}
Since $G_\alpha(0)(1-\alpha)=1$, this family describes the regime
$\alpha\uparrow1$ under the rough-signal cancellation condition.
We use the notation
\begin{equation*}
\mathcal Hf:=G'\ast f,
\qquad
\mathcal I_\alpha f:=G_\alpha(0)f+G_\alpha'\ast f,
\qquad
\Phi_\alpha:=G_\alpha(0)\phi_\alpha
=\frac{\phi_\alpha}{1-\alpha}.
\end{equation*}
For the equilibrium kernels supplied by
\cref{prop:selected_coboundary_equilibrium}, the kernel of $X^\alpha$ from
\eqref{eq:Ktilde_def} becomes
\begin{equation}
\label{eq:v2_observed_kernel_cancellation_curve}
\mathcal K_\alpha
=K+\mathcal I_\alpha\phi_\alpha
=K+(\mathrm{Id}+\mathcal H)\Phi_\alpha.
\end{equation}

\subsubsection{Functional setting}

To control the tails uniformly as $\alpha\uparrow1$, fix $p>0$ and introduce
the weighted space
\begin{equation*}
\mathcal B_p
:=
\left\{
h:\mathbb R_+^*\to\mathbb R:\
h\text{ is Borel},\quad
\norm{h}_{p}:=\sup_{t>0}(1+t)^p|h(t)|<\infty
\right\}.
\end{equation*}
As in \Cref{subsubsec:no_transient_impact}, let $n_0$ be the smallest integer
such that $\mathcal S_\nu^{n_0}K\in\mathcal B_0$, and assume further that
$\mathcal S_\nu^{n_0}K\in\mathcal B_p$. We work on the Banach extension of
$\mathcal B_p$ defined by
\begin{equation*}
\mathcal Y_K^p
:=
\mathcal B_p
\oplus\operatorname{span}\{\mathcal S_\nu^jK:0\le j<n_0\}.
\end{equation*}
We equip $\mathcal Y_K^p$ with the direct-sum norm analogous to
\eqref{eq:app_YK_direct_norm}, with $\mathcal B_p$ in place of $\mathcal B_0$.
As in \Cref{subsubsec:no_transient_impact}, the Abel operator $Q_\alpha$ is
well defined on $\mathcal Y_K^p$ and preserves it for every $\alpha<1$.
Under the rough-signal cancellation condition, $c_\alpha=\alpha$, and the
equation for the rescaled kernel is
\begin{equation}
\label{eq:v2_scaling_Q_resolvent}
(\mathrm{Id}-\alpha Q_\alpha\mathcal H)\Phi_\alpha
=Q_\alpha K.
\end{equation}

\subsection{Scaling limit of aggregate-position kernels and prices}

Formally, $Q_\alpha$ approaches $-\mathrm{Id}$ as $\alpha\uparrow1$. As such,
\eqref{eq:v2_scaling_Q_resolvent} suggests the limiting equation
$(\mathrm{Id}+\mathcal H)\Phi=-K$. The following proposition makes this limit
rigorous under suitable assumptions.

\begin{proposition}[Scaling limit of rescaled aggregate-position kernels]
\label{prop:scaling_leading}
Let $p>0$ and suppose
$\mathcal S_\nu^{n_0}K\in\mathcal B_p$. Assume that
\begin{equation*}
\norm{G'}_{L^1(\mathbb R_+)}<1,
\qquad
\int_0^\infty(1+t)^p|G'(t)|\,\diff t<\infty,
\end{equation*}
and
\begin{equation*}
\mathcal H(\mathcal S_\nu^jK)\in\mathcal B_p,
\qquad 0\le j<n_0.
\end{equation*}
For $\alpha\in[\alpha_0,1)$, let $\Phi_\alpha\in\mathcal Y_K^p$ satisfy
\eqref{eq:v2_scaling_Q_resolvent}. Assume the uniform polynomial tail estimate
\begin{equation}
\label{eq:v2_weighted_phi_bound}
\sup_{\alpha\in[\alpha_0,1)}
\norm{\Phi_\alpha}_{\mathcal Y_K^p}<\infty
\end{equation}
for some $\alpha_0\in[0,1)$. Then
\begin{equation}
\label{eq:Phi_alpha_limit}
\Phi_\alpha
\longrightarrow
\Phi:=-(\mathrm{Id}+\mathcal H)^{-1}K
\qquad\text{in }\mathcal Y_K^0,
\end{equation}
where $\Phi$ is the unique solution in $\mathcal Y_K^0$ of
\begin{equation}
\label{eq:v2_leading_volterra}
(\mathrm{Id}+\mathcal H)\Phi=-K.
\end{equation}
Moreover,
\begin{equation}
\label{eq:K_alpha_unscaled_limit}
\mathcal K_\alpha\longrightarrow0
\qquad\text{in }\mathcal Y_K^0.
\end{equation}
\end{proposition}

\begin{remark}[Tail conditions]
Consider the kernel $K(t)=t^a$, where
$a=H-1/2\in(-1/2,0)$, and the Gamma-distributed horizons studied in
\Cref{sec:tail_regularity}. Then
\begin{equation*}
\mathcal S_\nu^{n_0}K(t)\sim t^a
\qquad\text{as }t\to\infty.
\end{equation*}
Consequently, as soon as
$\mathcal S_\nu^{n_0}K\in\mathcal B_0$, it belongs to $\mathcal B_p$ for
every $0<p\leq-a$.
For the operator $\mathcal H$, the weighted-convolution estimate
\begin{equation*}
\norm{\mathcal Hf}_p
\leq
\left(\int_0^\infty(1+t)^p|G'(t)|\,\diff t\right)\norm{f}_p,
\qquad f\in\mathcal B_p,
\end{equation*}
shows that $\mathcal H$ preserves $\mathcal B_p$. For
$0\leq j<n_0$, the function $\mathcal S_\nu^jK$ is locally integrable and
satisfies $\mathcal S_\nu^jK(t)\sim t^a$ as $t\to\infty$. The preceding
estimate shows that
$\mathcal H(\mathcal S_\nu^jK)\in\mathcal B_p$ whenever
$|G'(t)|=O(t^{-q})$ as $t\to\infty$ for some $q>1$ with $p<q-1$.
\end{remark}

\needspace{6\baselineskip}
The following remark discusses the scope of the uniform tail condition
\eqref{eq:v2_weighted_phi_bound}, which requires the family $(\Phi_\alpha)$ to
retain polynomial tail control near the boundary.

\begin{remark}[Uniform tail condition]
The uniform estimate \eqref{eq:v2_weighted_phi_bound} is not automatic.
\Cref{prop:scaling_ou_uniform_tails} verifies it in an Ornstein--Uhlenbeck
benchmark with exponentially decaying transient impact. The argument extends
to positive Laplace mixtures whose decay rates are bounded away from zero.
Although these examples do not cover the singular kernels
studied in \Cref{sec:tail_regularity}, they show that this condition is
compatible with a nonzero transient component and a selected equilibrium
family.
\end{remark}

\begin{proposition}[Uniform tails in the Ornstein--Uhlenbeck case]
\label{prop:scaling_ou_uniform_tails}
Let
\begin{equation*}
K(t)=e^{-\gamma t},
\qquad
G(t)=1-\Theta+\Theta e^{-\kappa t},
\end{equation*}
where $\gamma,\kappa>0$ and $0\leq\Theta\leq1/2$. Then the corresponding
rescaled equilibrium kernels $\Phi_\alpha$ satisfy
\eqref{eq:v2_weighted_phi_bound} for every $p>0$.
\end{proposition}

The proofs of
\cref{prop:scaling_leading,prop:scaling_ou_uniform_tails} are given in
Appendix~\ref{app:scaling_proofs}. In particular,
\cref{prop:scaling_impact_bounds} and
\eqref{eq:v2_observed_kernel_cancellation_curve} give
\begin{equation*}
\mathcal K_\alpha
=K+(\mathrm{Id}+\mathcal H)\Phi_\alpha
\longrightarrow
K+(\mathrm{Id}+\mathcal H)\Phi
=0.
\end{equation*}
Through the Volterra representation \eqref{eq:Ktilde_volterra}, this kernel
convergence yields the following process-level limit.

\begin{corollary}[Observed price limit]
\label{cor:scaling_observed_price_limit}
In addition to the assumptions of \cref{prop:scaling_leading}, suppose
\begin{equation*}
\mathcal S_\nu^jK\in L^2_{\mathrm{loc}}(\mathbb R_+),
\qquad 0\le j<n_0,
\end{equation*}
and suppose the corresponding kernel family satisfies the process-level
conditions of \cref{prop:meanfield_eq}. Then, for every $T>0$,
\begin{equation*}
\sup_{0\le t\le T}\E\!\left[|X_t^\alpha|^2\right]
\longrightarrow0,
\qquad
\sup_{0\le t\le T}\E\!\left[|S_t^\alpha-M_t|^2\right]
\longrightarrow0.
\end{equation*}
In particular, $X_t^\alpha\to0$ and $S_t^\alpha\to M_t$ in $L^2(\Omega)$
for every fixed $t\ge0$.
\end{corollary}

\begin{proof}
Under the assumption on $(\mathcal S_\nu^j K)_{0 \leq j <n_0}$, the embedding
$\mathcal Y_K^0\hookrightarrow L^2(0,T)$ is continuous. Hence
\eqref{eq:K_alpha_unscaled_limit} implies
$\norm{\mathcal K_\alpha}_{L^2(0,T)}\to0$. By It\^o isometry,
\begin{equation*}
\sup_{0\le t\le T}\E\!\left[|X_t^\alpha|^2\right]
=
\sup_{0\le t\le T}\int_0^t|\mathcal K_\alpha(u)|^2\,\diff u
\le
\norm{\mathcal K_\alpha}_{L^2(0,T)}^2
\longrightarrow0.
\end{equation*}
Since $S^\alpha-M=X^\alpha$, the second convergence follows from the first.
\end{proof}

Thus, along the equilibrium family and under the fixed-kernel scaling, the
predictable signal-driven component vanishes and the observed price converges
to the martingale residual $M$. The limiting price therefore has no predictable
increments and retains the appropriate Hölder regularity.

\section{Rough forecast signals and heterogeneous horizons}
\label{sec:tail_regularity}

\subsection{Economic specification and scope}

We now specialize the common forecastable signal to a fractional-type kernel
at the origin. We write
\begin{equation*}
    a=H-\frac12,
    \qquad H\in\left(0,\frac12\right),
\end{equation*}
and assume that a fixed Borel representative of $K$ satisfies
\begin{equation}
\label{eq:v2_rough_kernel_specification}
    K(\tau)=\tau^a+K_s(\tau),
    \qquad K_s \in \mathcal C^1(\R_+^*),
    \qquad
    K_s(\tau)= o(\tau^a),
    \qquad
    K_s'(\tau)= o(\tau^{a-1}),
    \qquad \tau\downarrow0.
\end{equation}
The exponent $H<1/2$ describes the local roughness of the predictable signal
at the mid-frequency scale. Together with the standing decay assumption on
$K$, this is a reduced-form description of medium-term predictability with
fading shock loadings, rather than of fundamental factors driving long-run
trends.
We focus on $H<1/2$ under this mid-frequency interpretation. The complementary
case $H>1/2$ could be treated by the same methods, and the same mechanisms would
instead lower effective regularity towards $1/2$.\\

According to \cite{nguyen2026approximation}, Gamma mixtures are dense in the
probability densities on $\R_+$ belonging to $L^p$, for every
$1\leq p<\infty$. Since horizon averaging is linear in $\nu$,
we focus on a single Gamma component, with shape $\xi>0$ and rate $\beta>0$:
\begin{equation*}
    \nu(\diff\lambda)
    =
    \frac{\beta^\xi}{\Gamma(\xi)}
    \lambda^{\xi-1}e^{-\beta\lambda}\,\diff\lambda.
\end{equation*}
When $0<\xi<1$, the density is singular at zero, while
$\xi=1$ gives the exponential law and $\xi>1$ gives a density that vanishes at
zero and has an interior mode
\begin{equation*}
    \frac{\xi-1}{\beta}.
\end{equation*}
The last regime can represent some kind of latency rather than
an immediate response.
We further assume that
$K_s'\in L^\infty([r,\infty))$ for every $r>0$ and that
$G''\in L^1(\R_+)\cap L^\infty(\R_+)$.

\subsection{Local shape of the equilibrium kernel}

The relevant index is
\begin{equation*}
    n_0:=\min\{n\geq1:a+n\xi>0\}
    =\left\lfloor\frac{-a}{\xi}\right\rfloor+1,
\end{equation*}
which is precisely the smallest index such that
$\mathcal S_\nu^{n_0}K\in\mathcal B_0$.
When $a+m\xi$ is an integer, the following expansions may
contain logarithmic terms and require separate treatment. To avoid these
boundary cases, we impose the nonresonance condition
\begin{equation}
\label{eq:v2_gamma_nonresonance}
    a+m\xi\notin\mathbb Z,
    \qquad m=1,\ldots,n_0.
\end{equation}
These are precisely the exponents arising in the finite sequence of smoothing
iterates and its first regular element. In particular, the exponential case $\xi=1$ is
included because $a+1\in(1/2,1)$.

\begin{definition}[Gamma-regular kernel solution]
\label{def:admissible_equilibrium_kernel}
A solution $\phi_\alpha$ of the deterministic kernel equation
\eqref{eq:meanfield_alpha} is called Gamma-regular if it admits the
local expansion
\begin{equation}
\label{eq:v2_phi_expansion}
    \phi_\alpha
    =
    \frac{1}{1+\alpha G(0)}
    \left[
        -K
        +\sum_{m=1}^{n_0-1}
          (1-c_\alpha)c_\alpha^{m-1}\mathcal S_\nu^mK
        +r_\alpha
    \right],
\end{equation}
where
\begin{equation*}
    c_\alpha:=\frac{\alpha G(0)}{1+\alpha G(0)},
\end{equation*}
and $r_\alpha\in\mathcal C^1(\R_+^*)$ is bounded together with its
derivative on $[r,\infty)$ for every $r>0$ and satisfies
\begin{equation*}
    r_\alpha(\tau)=o(\tau^a),
    \qquad
    r_\alpha'(\tau)=o(\tau^{a-1}),
    \qquad \tau\downarrow0.
\end{equation*}
\end{definition}

Under the hypotheses of \cref{prop:selected_coboundary_equilibrium}, we assume
in addition that the $\mathcal Y_K^0$-solution is Gamma-regular.\footnote{This is
a mild additional regularity assumption on the selected equilibrium branch. It
requires $\mathcal C^1$ control only of the remainder in
\eqref{eq:v2_phi_expansion}, once the finitely many explicit Gamma-smoothed
terms have been separated.}

\begin{proposition}[Singularity of Gamma-regular solutions]
\label{prop:v2_phi_asymp}
For every Gamma-regular kernel solution,
\begin{equation*}
    \phi_\alpha(\tau)
    \sim
    -\frac{\tau^a}{1+\alpha G(0)}
    \qquad\text{as }\tau\downarrow0.
\end{equation*}
The same conclusion holds at the derivative level:
\begin{equation*}
    \phi_\alpha'(\tau)
    \sim
    -\frac{a\tau^{a-1}}{1+\alpha G(0)}.
\end{equation*}
Consequently, the critical H\"older exponent of the associated aggregate
position
\begin{equation*}
    \Pi_t^\alpha=\int_0^t\phi_\alpha(t-s)\,\diff B_s
\end{equation*}
is $H=a+1/2<1/2$.
\end{proposition}

The Gamma-regular equilibrium kernel inherits the leading singularity of $K$,
so equilibrium trading preserves the signal's local roughness at the level of
aggregate positions. The proof is given in Appendix~\ref{app:regularity_proofs}.

\subsection{Away from the rough-signal cancellation condition}

By \eqref{eq:Ktilde_volterra},
\begin{equation*}
    X_t^\alpha
    =\int_0^t\mathcal K_\alpha(t-s)\,\diff B_s,
    \qquad
    S_t^\alpha=M_t+X_t^\alpha.
\end{equation*}
We first consider parameters away from the rough-signal cancellation
condition $G(0)(1-\alpha)=1$.

\begin{proposition}[Signal-driven regularity without cancellation]
\label{prop:v2_regularity_alpha}
Let $\phi_\alpha$ be Gamma-regular and suppose
$1-(1-\alpha)G(0)\neq0$. Then
\begin{align*}
    \mathcal K_\alpha(\tau)
    &\sim
    \frac{1-(1-\alpha)G(0)}{1+\alpha G(0)}\,\tau^a,
    \\
    \mathcal K_\alpha'(\tau)
    &\sim
    \frac{1-(1-\alpha)G(0)}{1+\alpha G(0)}\,a\tau^{a-1},
    \qquad \tau\downarrow0.
\end{align*}
Consequently, the critical H\"older exponent of $X^\alpha$ is
$H=a+1/2$.
\end{proposition}
Without cancellation, the leading singularity of $K$ survives in
$\mathcal K_\alpha$, and the observed price retains the local roughness
$H<1/2$ of the original signal. This is incompatible with a Brownian-type regularity and motivates the analysis
of the cancellation regime in the next subsection.
% This is incompatible with the regularity of a
% martingale with the assumed Brownian-type regularity and motivates the analysis
% of the cancellation regime in the next subsection.

\subsection{Under the rough-signal cancellation condition}

This subsection concerns the cancellation regime
\begin{equation*}
    G(0)(1-\alpha)=1.
\end{equation*}
Under this condition, the leading singularity inherited from $K$ cancels, and the local behavior of $\mathcal K_\alpha$ is
governed by the competition between horizon smoothing and the Volterra-impact terms.

\begin{proposition}[Signal-driven regularity under cancellation]
\label{prop:v2_bottleneck}
Let $\phi_\alpha$ be Gamma-regular and suppose that $G(0)(1-\alpha)=1$.
We write $\ell_\alpha:=\lim_{\tau\downarrow0}\mathcal K_\alpha(\tau)$,
which exists whenever $H+\xi>1/2$. The following assertions hold.
\begin{itemize}
    \item If $0<\xi<1$, then $H+\xi\neq1/2$ by
    \eqref{eq:v2_gamma_nonresonance}, and there is a nonzero constant
    $C_\alpha$ such that
    \begin{equation*}
        \mathcal K_\alpha(\tau)
        =
        \begin{cases}
            C_\alpha\tau^{H+\xi-1/2}+o(\tau^{H+\xi-1/2}),
                & H+\xi<1/2,\\
            \ell_\alpha+C_\alpha\tau^{H+\xi-1/2}+o(\tau^{H+\xi-1/2}),
                & H+\xi>1/2,
        \end{cases}
    \end{equation*}
    The corresponding differentiated expansion also holds for
    $\mathcal K_\alpha'$. Hence the critical H\"older exponent of
    $X^\alpha$ is $\rho_X^*$, where
    \begin{equation*}
        \rho_X^*:=
        \begin{cases}
            H+\xi, & H+\xi<1/2,\\
            1/2, & H+\xi>1/2\text{ and }\ell_\alpha\neq0,\\
            \min(1,H+\xi), & H+\xi>1/2\text{ and }\ell_\alpha=0.
        \end{cases}
    \end{equation*}

    \item If $\xi\geq1$,
    \begin{equation*}
        \mathcal K_\alpha(\tau)=\ell_\alpha+O(\tau^{a+1}),
        \qquad
        \mathcal K_\alpha'(\tau)=O(\tau^a).
    \end{equation*}
    If $\ell_\alpha\neq0$, then the critical H\"older exponent of
    $X^\alpha$ is $1/2$. If $\ell_\alpha=0$, then $X^\alpha$ has
    locally $\rho$-H\"older paths almost surely for every $\rho<1$.
\end{itemize}
\end{proposition}

The relevant threshold is $\xi=\frac12-H$. If
$0<\xi<\frac12-H$, cancellation raises the local regularity of both
$X^\alpha$ and $S^\alpha$ from $H$ to $H+\xi$, but the observed price remains
rougher than the martingale component. Above the threshold, the critical
H\"older exponent of $X^\alpha$ is $1/2$ when $\ell_\alpha\neq0$. When
$\ell_\alpha=0$ and $0<\xi<1$, it is $\min(1,H+\xi)$, whereas for $\xi\geq1$,
$X^\alpha$ has locally $\rho$-H\"older paths almost surely for every $\rho<1$. Whenever $\xi > \frac{1}{2} - H$, the paths of the
full price are almost surely locally $\rho$-H\"older for every $\rho<1/2$.\\

\Cref{prop:v2_bottleneck} therefore provides a natural condition on the
distribution of forecast horizons for the observed price to be compatible with
an efficient market, in the local regularity sense considered here. Under
rough-signal cancellation, this occurs when $\xi>\frac12-H$. The horizon
distribution then assigns sufficiently little mass to arbitrarily short horizons
for horizon averaging to make the signal-driven component at least as regular
as the martingale component. This includes Gamma laws with
$\frac12-H<\xi<1$, whose density remains singular at zero, as well as the
exponential law $\xi=1$ and Gamma laws with $\xi>1$.
When $0<\xi<\frac12-H$, by contrast, the distribution is too concentrated near
zero. Agents then react predominantly to short-horizon forecasts that closely
track the rough predictable signal. Aggregating their positions consequently
provides insufficient smoothing, and part of the signal roughness survives in
prices. Cancellation still raises the local regularity from $H$ to $H+\xi$,
but the observed price remains rougher than its martingale component because
$H+\xi<1/2$. Finally, the rate $\beta$ changes the overall horizon scale
without changing the density's power at zero, and therefore does not affect
this local efficiency criterion.

\section*{Acknowledgments}

The authors gratefully acknowledge support from the ILB Chair
\emph{Artificial Intelligence and Quantitative Methods for Finance} at
University Paris Dauphine--PSL.

\bibliographystyle{plainnat}
\bibliography{references}

\appendix

\makeatletter
\renewcommand{\@seccntformat}[1]{%
  \ifstrequal{#1}{section}
    {Appendix~\csname the#1\endcsname.\quad}
    {\csname the#1\endcsname\quad}}
\makeatother

\section{Existence and uniqueness in \texorpdfstring{$\mathcal Y_K^0$}{Y K 0}}
\label{app:coboundary_existence}

\subsection{Resolvent on \texorpdfstring{$\mathcal B_0$}{B0}}
\label{subsec:app_resolvent_B0}

\begin{lemma}
\label{lem:app_shift_tail}
The operator $\mathcal S_\nu$ maps $\mathcal B_0$ continuously into itself,
with
\begin{equation*}
\norm{\mathcal S_\nu f}_{\mathcal B_0}
\le\norm{f}_{\mathcal B_0}.
\end{equation*}
Moreover,
\begin{equation*}
\norm{\mathcal S_\nu^n f}_{\mathcal B_0}\longrightarrow0,
\qquad n\to\infty,
\end{equation*}
for every $f\in\mathcal B_0$.
\end{lemma}

\iffalse
\begin{proof}
Boundedness and Borel measurability are preserved by integration against the
probability measure $\nu$, and
\begin{equation*}
|\mathcal S_\nu f(t)|
\le\int_{(0,\infty)}\norm{f}_{\mathcal B_0}\,\nu(d\lambda)
=\norm{f}_{\mathcal B_0}.
\end{equation*}
If $f(t)\to0$, then for every $\varepsilon>0$ there is $R$ such that
$|f(x)|\le\varepsilon$ for $x\ge R$. Hence
$|\mathcal S_\nu f(t)|\le\varepsilon$ for $t\ge R$, so
$\mathcal S_\nu f$ also vanishes at infinity.

Let $\Lambda_1,\Lambda_2,\ldots$ be independent with common law $\nu$. Then
\begin{equation*}
\mathcal S_\nu^n f(t)
=\E\bigl[f(t+\Lambda_1+\cdots+\Lambda_n)\bigr].
\end{equation*}
Because $(\Lambda_i)_{i\geq1}$ are i.i.d. and $\nu$ is supported on
$(0,\infty)$, the partial sums tend to infinity almost surely. For any $R>0$,
\begin{align*}
\sup_{t>0}
\E\bigl[|f(t+\Lambda_1+\cdots+\Lambda_n)|\bigr]
&\le
\sup_{x\ge R}|f(x)|
\\
&\quad+
\norm{f}_{\mathcal B_0}
\Prob(\Lambda_1+\cdots+\Lambda_n<R).
\end{align*}
Taking $n\to\infty$ and subsequently $R\to\infty$ proves
$\norm{\mathcal S_\nu^n f}_{\mathcal B_0}\to0$.
\end{proof}
\fi

\begin{proof}
Measurability is preserved under integration, while
$|\mathcal S_\nu f(t)|\leq\sup_{x\geq t}|f(x)|$. Thus $\mathcal S_\nu$ is a
contraction on $\mathcal B_0$ and preserves decay at infinity.
Let $S_n:=\Lambda_1+\cdots+\Lambda_n$ for i.i.d.\ random variables with law
$\nu$. Since $\nu$ is supported on $(0,\infty)$, $S_n\to\infty$ almost surely,
and, for every $R>0$,
\begin{equation*}
\norm{\mathcal S_\nu^n f}_{\mathcal B_0}
\leq
\sup_{x\geq R}|f(x)|
+\norm{f}_{\mathcal B_0}\Prob(S_n<R).
\end{equation*}
Letting first $n$ and then $R$ tend to infinity proves the claim.
\end{proof}

For $c\in[0,1]$, define formally
\begin{equation*}
Q_cf:=\sum_{n\ge0}c^n\mathcal S_\nu^n\mathcal D_\nu f,
\qquad
\mathcal D_\nu=\mathcal S_\nu-\mathrm{Id}.
\end{equation*}

\begin{lemma}
\label{lem:app_Qc_B0}
For every $f\in\mathcal B_0$, the series defining $Q_cf$ converges in
$\mathcal B_0$ and
\begin{equation*}
\norm{Q_cf}_{\mathcal B_0}\le2\norm{f}_{\mathcal B_0}.
\end{equation*}
It satisfies
\begin{equation}
\label{eq:app_Qc_identity_B0}
(\mathrm{Id}-c\mathcal S_\nu)Q_cf=\mathcal D_\nu f.
\end{equation}
At $c=1$, $Q_1=-\mathrm{Id}$, and for every fixed $f\in\mathcal B_0$,
\begin{equation*}
Q_cf\longrightarrow-f
\qquad\text{in }\mathcal B_0
\qquad\text{as }c\uparrow1.
\end{equation*}
\end{lemma}

\iffalse
\begin{proof}
For $0\le c<1$, summation by parts gives
\begin{align*}
\sum_{n=0}^N c^n\mathcal S_\nu^n\mathcal D_\nu f
&=-f
+(1-c)\sum_{m=1}^N c^{m-1}\mathcal S_\nu^m f
+c^N\mathcal S_\nu^{N+1}f.
\end{align*}
The middle term converges absolutely in $\mathcal B_0$, and the last term
vanishes by \cref{lem:app_shift_tail}. Thus
\begin{equation*}
Q_cf
=-f+(1-c)\sum_{m=1}^\infty c^{m-1}\mathcal S_\nu^m f,
\end{equation*}
which gives the stated bound. At $c=1$, the partial sums telescope:
\begin{equation*}
\sum_{n=0}^N\mathcal S_\nu^n\mathcal D_\nu f
=\mathcal S_\nu^{N+1}f-f
\longrightarrow-f.
\end{equation*}
Equation~\eqref{eq:app_Qc_identity_B0} follows from the convergent series.
Finally, the second term in \eqref{eq:app_Qc_Abel_B0} tends to zero as
$c\uparrow1$ because it is the Abel mean of the norm-null sequence
$(\mathcal S_\nu^m f)_{m\ge1}$.
\end{proof}
\fi

\begin{proof}
For $c<1$, summation by parts and the contraction property of
$\mathcal S_\nu$ give
\begin{equation*}
Q_cf
=-f+(1-c)\sum_{m=1}^\infty c^{m-1}\mathcal S_\nu^m f,
\end{equation*}
which proves convergence and the stated bound. At $c=1$, the partial sums are
$\mathcal S_\nu^{N+1}f-f$ and converge to $-f$ by
\cref{lem:app_shift_tail}. Applying the bounded operator
$\mathrm{Id}-c\mathcal S_\nu$ to the convergent partial sums gives
\eqref{eq:app_Qc_identity_B0}. Finally,
$\mathcal S_\nu^m f\to0$ in $\mathcal B_0$, so its Abel means converge to
zero, proving that $Q_cf\to-f$ as $c\uparrow1$.
\end{proof}

\subsection{A finite extension of \texorpdfstring{$\mathcal B_0$}{B0}}
\label{subsec:app_existence_theorem}

Assume that $\mathcal S_\nu^nK\in\mathcal B_0$ for some $n\geq0$, and let
$n_0$ be the smallest such integer.

\begin{lemma}
\label{lem:app_YK_direct_sum}
The following finite extension is the direct sum
\begin{equation*}
\mathcal Y_K^0
=\mathcal B_0\oplus
\mathrm{span}\{K,\mathcal S_\nu K,\ldots,\mathcal S_\nu^{n_0-1}K\}.
\end{equation*}
Every $f\in\mathcal Y_K^0$ therefore has a unique representation
$f=b+\sum_{j=0}^{n_0-1}a_j\mathcal S_\nu^jK$, with
$b\in\mathcal B_0$, $a_j\in\R$ for $0\leq j<n_0$, and the norm
\begin{equation}
\label{eq:app_YK_direct_norm}
\norm{f}_{\mathcal Y_K^0}
:=\norm{b}_{\mathcal B_0}+\sum_{j=0}^{n_0-1}|a_j|
\end{equation}
makes $\mathcal Y_K^0$ a Banach space.
\end{lemma}

\iffalse
\begin{proof}
The case $n_0=0$ is immediate. Suppose $n_0\geq1$,
$\sum_{j=0}^{n_0-1}a_jv_j\in\mathcal B_0$, and let $j_0$ be the
smallest index for which $a_{j_0}\neq0$. Applying
$\mathcal S_\nu^{n_0-1-j_0}$ shows that
\begin{equation*}
    a_{j_0}v_{n_0-1}\in\mathcal B_0,
\end{equation*}
because every term with larger index has reached $\mathcal B_0$ and
$\mathcal S_\nu$ preserves $\mathcal B_0$. This contradicts the minimality of
$n_0$and the sum is therefore direct. 
\end{proof}
\fi

\begin{proof}
The case $n_0=0$ is immediate. Otherwise, suppose that
\begin{equation*}
\sum_{j=0}^{n_0-1}a_j\mathcal S_\nu^jK\in\mathcal B_0,
\end{equation*}
and let $j_0$ be the first index with $a_{j_0}\neq0$. After applying
$\mathcal S_\nu^{n_0-1-j_0}$, every term except
$a_{j_0}\mathcal S_\nu^{n_0-1}K$ belongs to $\mathcal B_0$. Hence
$\mathcal S_\nu^{n_0-1}K\in\mathcal B_0$, contradicting the minimality of
$n_0$. The sum is therefore direct and the norm
\eqref{eq:app_YK_direct_norm} makes it complete.
\end{proof}

\begin{lemma}[Abel resolvent extension]
\label{lem:app_Qc_extension}
For every $c\in[0,1]$, the series
\begin{equation*}
Q_cf:=\sum_{n\ge0}c^n\mathcal S_\nu^n\mathcal D_\nu f
\end{equation*}
converges for every $f\in\mathcal Y_K^0$ and thereby extends the operator
$Q_c$ constructed on $\mathcal B_0$ in \cref{lem:app_Qc_B0} to a bounded
operator on $\mathcal Y_K^0$. It satisfies
\begin{equation}
\label{eq:app_Qc_extension_identity}
(\mathrm{Id}-c\mathcal S_\nu)Q_c=\mathcal D_\nu
\qquad\text{on }\mathcal Y_K^0.
\end{equation}
Moreover, $Q_1=-\mathrm{Id}$ and
\begin{equation*}
Q_cf\longrightarrow-f
\qquad\text{in }\mathcal Y_K^0
\qquad\text{as }c\uparrow1
\end{equation*}
for every fixed $f\in\mathcal Y_K^0$.
\end{lemma}

\iffalse
\begin{proof}
Let $Q_c^{(0)}$ denote the operator on $\mathcal B_0$ constructed in
\cref{lem:app_Qc_B0}. Since $\mathcal S_\nu$ preserves $\mathcal B_0$ and
$\mathcal S_\nu^{n_0+j}K
=\mathcal S_\nu^j(\mathcal S_\nu^{n_0}K)\in\mathcal B_0$, one has
\begin{equation*}
\mathcal S_\nu^{n_0}\mathcal Y_K^0\subset\mathcal B_0.
\end{equation*}
Moreover, the direct-sum norm shows that
$\mathcal S_\nu:\mathcal Y_K^0\to\mathcal Y_K^0$ and
$\mathcal S_\nu^{n_0}:\mathcal Y_K^0\to\mathcal B_0$ are bounded. Hence the
formal tail decomposition
\begin{equation*}
Q_c
:=
\sum_{n=0}^{n_0-1}c^n\mathcal S_\nu^n\mathcal D_\nu
+c^{n_0}Q_c^{(0)}\mathcal S_\nu^{n_0}
\end{equation*}
extends $Q_c^{(0)}$ to a bounded operator on $\mathcal Y_K^0$.

Applying $\mathrm{Id}-c\mathcal S_\nu$ to the finite-prefix identity yields
\begin{align*}
(\mathrm{Id}-c\mathcal S_\nu)Q_c
&=\mathcal D_\nu
-c^{n_0}\mathcal S_\nu^{n_0}\mathcal D_\nu
+c^{n_0}\mathcal D_\nu\mathcal S_\nu^{n_0}
=\mathcal D_\nu,
\end{align*}
which proves \eqref{eq:app_Qc_extension_identity}. At $c=1$,
$Q_1^{(0)}=-\mathrm{Id}$ and, for every $f\in\mathcal Y_K^0$, telescoping
gives
\begin{equation*}
Q_1f
=\bigl(\mathcal S_\nu^{n_0}f-f\bigr)
-\mathcal S_\nu^{n_0}f
=-f.
\end{equation*}
As $c\uparrow1$, the finite prefix converges to
$\mathcal S_\nu^{n_0}f-f$, while \cref{lem:app_Qc_B0}, applied to
$\mathcal S_\nu^{n_0}f\in\mathcal B_0$, gives
$c^{n_0}Q_c^{(0)}\mathcal S_\nu^{n_0}f\to
-\mathcal S_\nu^{n_0}f$. Hence $Q_cf\to-f$ in $\mathcal Y_K^0$.
\end{proof}
\fi

\begin{proof}
By \cref{lem:app_YK_direct_sum} and the definition of $n_0$,
$\mathcal S_\nu$ is bounded on
$\mathcal Y_K^0$, while $\mathcal S_\nu^{n_0}$ maps $\mathcal Y_K^0$
boundedly into $\mathcal B_0$. Writing $Q_c^{(0)}$ for the resolvent on
$\mathcal B_0$, reindexing the tail from $n_0$ gives, by definition,
\begin{equation*}
Q_cf
=
\sum_{n=0}^{n_0-1}c^n\mathcal S_\nu^n\mathcal D_\nu f
+c^{n_0}Q_c^{(0)}\mathcal S_\nu^{n_0}f.
\end{equation*}
Since $Q_c^{(0)}$ is bounded on $\mathcal B_0$, this proves convergence of
the series in $\mathcal Y_K^0$ and boundedness of the extension.

Applying $\mathrm{Id}-c\mathcal S_\nu$ and using
\eqref{eq:app_Qc_identity_B0} yields
\begin{equation*}
(\mathrm{Id}-c\mathcal S_\nu)Q_cf
=\mathcal D_\nu f
+c^{n_0}
\bigl(
\mathcal D_\nu\mathcal S_\nu^{n_0}
-\mathcal S_\nu^{n_0}\mathcal D_\nu
\bigr)f
=\mathcal D_\nu f.
\end{equation*}
At $c=1$, the finite sum equals $\mathcal S_\nu^{n_0}f-f$ by telescoping,
while the tail equals $-\mathcal S_\nu^{n_0}f$, so $Q_1=-\mathrm{Id}$. The same decomposition,
together with $Q_c^{(0)}g\to-g$ in $\mathcal B_0$, gives
$Q_cf\to-f$ in $\mathcal Y_K^0$ as $c\uparrow1$.
\end{proof}

\subsection{Transient-impact perturbation}
\label{subsec:app_transient_impact_perturbation}

Since $G'\in L^1(\mathbb R_+)$, the Volterra operator $\mathcal H$ maps
$\mathcal B_0$ continuously into itself and
\begin{equation}
\label{eq:app_H_B0_bound}
\norm{\mathcal Hf}_{\mathcal B_0}
\le\norm{G'}_1\norm{f}_{\mathcal B_0}.
\end{equation}
Together with \cref{lem:app_Qc_B0}, this also shows that
$Q_{c_\alpha}\mathcal H:\mathcal B_0\to\mathcal B_0$ is bounded.
We first seek a solution of \eqref{eq:coboundary_fixed_point} in the
affine class
\begin{equation*}
\psi_\alpha=Q_{c_\alpha}K+\eta_\alpha,
\qquad \eta_\alpha\in\mathcal B_0.
\end{equation*}
Using \eqref{eq:app_Qc_extension_identity}, the coboundary equation holds if
\begin{equation}
\label{eq:app_correction_preimage}
(\mathrm{Id}-c_\alpha\mathcal S_\nu)\eta_\alpha
=d_\alpha\mathcal D_\nu\mathcal H
\bigl(Q_{c_\alpha}K+\eta_\alpha\bigr).
\end{equation}
Injectivity of $\mathrm{Id}-c_\alpha\mathcal S_\nu$ on $\mathcal B_0$ makes
\eqref{eq:app_correction_preimage} equivalent to
\begin{equation}
\label{eq:app_correction_fixed_point}
\eta_\alpha
=d_\alpha Q_{c_\alpha}\mathcal H
\bigl(Q_{c_\alpha}K+\eta_\alpha\bigr).
\end{equation}

The fixed-point argument used in \cref{thm:app_finite_extension_existence}
below requires the right-hand side of
\eqref{eq:app_correction_fixed_point} to lie in $\mathcal B_0$, which follows
from the next lemma.

\begin{lemma}[Volterra mapping on $\mathcal Y_K^0$]
\label{lem:app_H_extension}
The operator $\mathcal H$ maps $\mathcal Y_K^0$ boundedly into
$\mathcal B_0$.
\end{lemma}

\begin{proof}
By \cref{lem:app_YK_direct_sum}, it suffices to show
that $\mathcal H\mathcal S_\nu^jK\in\mathcal B_0$ for $0\leq j<n_0$.
For every $T>0$,
\begin{align*}
\int_0^T|\mathcal S_\nu^jK(t)|\,\diff t
&\leq\sup_{r\geq0}\int_r^{r+T}|K(u)|\,\diff u<\infty,
\\
|\mathcal S_\nu^jK(t)|
&\leq\sup_{u\geq t}|K(u)|\longrightarrow0.
\end{align*}
Let $R>0$ be such that
$(\mathcal S_\nu^jK)\mathbf 1_{[R,\infty)}\in\mathcal B_0$. Its image under
$\mathcal H$ belongs to $\mathcal B_0$ by \eqref{eq:app_H_B0_bound}, while
the image of
$(\mathcal S_\nu^jK)\mathbf 1_{(0,R)}$ also belongs to $\mathcal B_0$ because
it is the convolution of a compactly supported $L^1$-function with $G'$, a
bounded function vanishing at infinity.
\end{proof}

We can now apply the contraction principle to
\eqref{eq:app_correction_fixed_point} and prove uniqueness in
$\mathcal Y_K^0$.

\begin{theorem}[Kernel construction and uniqueness in $\mathcal Y_K^0$]
\label{thm:app_finite_extension_existence}
Suppose that
\begin{equation}
\label{eq:app_smallness}
|d_\alpha|\,
\norm{Q_{c_\alpha}\mathcal H}_{\mathcal B_0\to\mathcal B_0}<1.
\end{equation}
Then \eqref{eq:app_correction_fixed_point} admits a solution
$\eta_\alpha\in\mathcal B_0$, and
\begin{equation*}
\phi_\alpha
:=\frac{Q_{c_\alpha}K+\eta_\alpha}{1+\alpha G(0)}
\in\mathcal Y_K^0
\end{equation*}
is the unique solution in $\mathcal Y_K^0$ of the kernel equation
\eqref{eq:meanfield_alpha}.
\end{theorem}

\iffalse
\begin{proof}
The assumption $\mathcal H Q_{c_\alpha}K\in\mathcal B_0$ ensures that
$d_\alpha Q_{c_\alpha}\mathcal H Q_{c_\alpha}K\in\mathcal B_0$, while
\eqref{eq:app_smallness} makes the right-hand side of
\eqref{eq:app_correction_fixed_point} an affine contraction on
$\mathcal B_0$. Hence it has the unique fixed point
\begin{equation*}
\eta_\alpha
=\sum_{m\ge0}d_\alpha^{m+1}
(Q_{c_\alpha}\mathcal H)^m
Q_{c_\alpha}\mathcal H Q_{c_\alpha}K,
\end{equation*}
with convergence in $\mathcal B_0$. $\psi_\alpha$ then solves the
coboundary equation, and the rescaling defining $\phi_\alpha$ gives
\eqref{eq:meanfield_alpha}.
\end{proof}
\fi

\begin{proof}
By \cref{lem:app_H_extension},
$\mathcal H:\mathcal Y_K^0\to\mathcal B_0$ is bounded. In particular,
$\mathcal H Q_{c_\alpha}K\in\mathcal B_0$. On $\mathcal B_0$, set
\begin{equation*}
T_\alpha\eta
:=d_\alpha Q_{c_\alpha}\mathcal H(Q_{c_\alpha}K+\eta).
\end{equation*}
The preceding mapping property and \eqref{eq:app_smallness} show that
$T_\alpha$ defines an affine map on
$\mathcal B_0$ whose linear part is a strict contraction. Its unique fixed
point has the Neumann expansion
\begin{equation*}
\eta_\alpha
=\sum_{m\geq0}d_\alpha^{m+1}
(Q_{c_\alpha}\mathcal H)^m
Q_{c_\alpha}\mathcal H Q_{c_\alpha}K,
\end{equation*}
with convergence in $\mathcal B_0$. Set
$\psi_\alpha:=Q_{c_\alpha}K+\eta_\alpha$. The preceding equivalence gives its
coboundary equation.\\

It remains to prove uniqueness in $\mathcal Y_K^0$. First,
$\mathrm{Id}-c_\alpha\mathcal S_\nu$ is injective on $\mathcal Y_K^0$.
Indeed, if $y=c_\alpha\mathcal S_\nu y$, then
$y=c_\alpha^{n_0}\mathcal S_\nu^{n_0}y\in\mathcal B_0$, since
$\mathcal S_\nu^{n_0}K\in\mathcal B_0$. This implies $y=0$, as
$\norm{c_\alpha\mathcal S_\nu}_{\mathcal B_0\to\mathcal B_0}<1$.\\

Let $\widehat\psi\in\mathcal Y_K^0$ be another solution and put
$u:=\widehat\psi-\psi_\alpha$. Then
\begin{equation*}
(\mathrm{Id}-c_\alpha\mathcal S_\nu)u
=d_\alpha\mathcal D_\nu\mathcal Hu.
\end{equation*}
Since $\mathcal Hu\in\mathcal B_0$, the Abel identity gives
\begin{equation*}
(\mathrm{Id}-c_\alpha\mathcal S_\nu)
\bigl(u-d_\alpha Q_{c_\alpha}\mathcal Hu\bigr)=0.
\end{equation*}
Injectivity of $\mathrm{Id}-c_\alpha\mathcal S_\nu$ therefore yields
\begin{equation*}
u=d_\alpha Q_{c_\alpha}\mathcal Hu\in\mathcal B_0.
\end{equation*}
Consequently,
\begin{equation*}
\norm{u}_{\mathcal B_0}
\leq
|d_\alpha|\,
\norm{Q_{c_\alpha}\mathcal H}_{\mathcal B_0\to\mathcal B_0}
\norm{u}_{\mathcal B_0},
\end{equation*}
and \eqref{eq:app_smallness} gives $u=0$.
\end{proof}

To turn the unique kernel in $\mathcal Y_K^0$ constructed in
\cref{thm:app_finite_extension_existence} into a Gaussian--Volterra
equilibrium through \cref{prop:meanfield_eq}, we establish the following
$\nu$-admissibility result for elements of $\mathcal Y_K^0$.

\begin{lemma}[Uniform local integrability]
\label{lem:app_uloc_admissibility}
For a measurable function $f$, set
\begin{equation*}
    N_T(f):=\sup_{r\geq0}\norm{f}_{L^2(r,r+T)}.
\end{equation*}
If $N_T(f)<\infty$ for every $T>0$, then $f$ is $\nu$-admissible. Moreover,
$N_T(f)<\infty$ for $f=K$, for every $f\in\mathcal Y_K^0$, and for
$\mathcal Hf$ whenever $f\in\mathcal Y_K^0$.
\end{lemma}

\iffalse
\begin{proof}
The assumptions $K\in L^2_{\mathrm{loc}}$ and $K(t)\to0$ imply
$N_T(K)<\infty$, and
$N_T(b)\leq\sqrt{T}\norm{b}_{\mathcal B_0}$ for every
$b\in\mathcal B_0$. For every $f\in\mathcal Y_K^0$, one has
$N_T(f)<\infty$, as Minkowski's inequality gives
\begin{equation*}
    N_T(\mathcal S_\nu f)\leq N_T(f).
\end{equation*}
According to Minkowski's inequality, one moreover has
\begin{equation*}
    N_T(\mathcal Hf)\leq\norm{G'}_1N_T(f).
\end{equation*}
Finally,
\begin{equation*}
\int_{(0,\infty)}
\norm{\Delta_\lambda f}_{L^2(0,T)}\,\nu(\diff\lambda)
\leq 2N_T(f),
\end{equation*}
which proves $\nu$-admissibility.
\end{proof}
\fi

\begin{proof}
For fixed $T$, local square integrability and decay at infinity give
$N_T(K)<\infty$.
Moreover,
$N_T(b)\leq\sqrt{T}\norm{b}_{\mathcal B_0}$ for $b\in\mathcal B_0$.
For every $f\in\mathcal Y_K^0$, Minkowski's inequality gives
\begin{equation*}
N_T(\mathcal S_\nu f)\leq N_T(f).
\end{equation*}
The preceding statements then yield $N_T(f)<\infty$ for every
$f\in\mathcal Y_K^0$.\\

For $f\in\mathcal Y_K^0$, Minkowski's inequality also gives
\begin{equation*}
N_T(\mathcal Hf)\leq\norm{G'}_1N_T(f).
\end{equation*}
Finally,
\begin{equation*}
\int_{(0,\infty)}
\norm{\Delta_\lambda f}_{L^2(0,T)}\,\nu(\diff\lambda)
\leq2N_T(f),
\end{equation*}
which is precisely the required $\nu$-admissibility bound.
\end{proof}

\iffalse
\begin{remark}[Proof of \cref{prop:selected_coboundary_equilibrium}]
The hypotheses of \cref{prop:selected_coboundary_equilibrium} imply those of
\cref{thm:app_finite_extension_existence}. Substituting the series for
$\eta_\alpha$ from its proof into
$\phi_\alpha=(Q_{c_\alpha}K+\eta_\alpha)/(1+\alpha G(0))$ gives
\eqref{eq:selected_equilibrium_expansion}. Since
$\phi_\alpha\in\mathcal Y_K^0$, \cref{lem:app_uloc_admissibility} shows that
$\phi_\alpha\in L^2_{\mathrm{loc}}$ and that $K$ and
$\mathcal I\phi_\alpha$ are $\nu$-admissible. The equilibrium conclusion
follows from \cref{prop:meanfield_eq}.
\end{remark}
\fi

\begin{proof}[Proof of \Cref{prop:selected_coboundary_equilibrium}]
\Cref{thm:app_finite_extension_existence} and its Neumann expansion give
\eqref{eq:selected_equilibrium_expansion} and uniqueness in $\mathcal Y_K^0$. Since
$\phi_\alpha\in\mathcal Y_K^0$, \cref{lem:app_uloc_admissibility} gives finite
$N_T$ for $\phi_\alpha$ and $\mathcal H\phi_\alpha$, hence also for
$\mathcal I\phi_\alpha$. Thus $\phi_\alpha\in L^2_{\mathrm{loc}}$, while
$K$, $\phi_\alpha$, and $\mathcal I\phi_\alpha$ are $\nu$-admissible. The conclusion follows
from \cref{prop:meanfield_eq}.
\end{proof}

\section{Gamma smoothing and Gaussian--Volterra regularity}
\label{app:auxiliary_regularity}

Throughout this appendix, $\nu$ denotes the Gamma distribution with shape
$\xi>0$ and rate $\beta>0$:
\begin{equation*}
    \nu(\diff\lambda)
    =
    \frac{\beta^\xi}{\Gamma(\xi)}
    \lambda^{\xi-1}e^{-\beta\lambda}\,\diff\lambda.
\end{equation*}
Since the sum of $m$ independent Gamma$(\xi,\beta)$ variables is
Gamma$(m\xi,\beta)$, for every $m\geq1$,
\begin{equation*}
    \mathcal S_\nu^m f(\tau)
    =
    \frac{\beta^{m\xi}}{\Gamma(m\xi)}
    \int_0^\infty f(\tau+\lambda)
    \lambda^{m\xi-1}e^{-\beta\lambda}\,\diff\lambda.
\end{equation*}

\subsection{Gamma smoothing of powers}

For $a<0$, write
\begin{equation*}
    p_a(\tau):=\tau^a,
    \qquad
    \Theta_a:=\mathcal S_\nu p_a.
\end{equation*}

\begin{proposition}[Gamma asymptotics]
\label{prop:gamma_asymp}
Let $a<0$ and set $b:=a+\xi$. To avoid a separate treatment of the logarithmic
terms that may arise when $b\in\mathbb Z$, we assume $b\notin\mathbb Z$. Then
\begin{equation*}
    \Theta_a(\tau)
    =
    \beta^\xi\tau^b U(\xi,b+1,\beta\tau),
\end{equation*}
where $U$ is the confluent hypergeometric function. Set
\begin{equation*}
    C_{a,\xi}
    :=
    \beta^\xi\frac{\Gamma(-b)}{\Gamma(-a)}
    \neq0.
\end{equation*}
When $b>0$,
\begin{equation*}
    \Theta_a(0)=\beta^{-a}\frac{\Gamma(b)}{\Gamma(\xi)}.
\end{equation*}
As $\tau\downarrow0$, the following Taylor expansions hold.
\begin{itemize}
    \item If $b<0$,
    \begin{align*}
        \Theta_a(\tau)
        &=C_{a,\xi}\tau^b+o(\tau^b),
        \\
        \Theta_a'(\tau)
        &=bC_{a,\xi}\tau^{b-1}+o(\tau^{b-1}).
    \end{align*}

    \item If $0<b<1$,
    \begin{align*}
        \Theta_a(\tau)
        &=\Theta_a(0)+C_{a,\xi}\tau^b+o(\tau^b),
        \\
        \Theta_a'(\tau)
        &=bC_{a,\xi}\tau^{b-1}+o(\tau^{b-1}).
    \end{align*}

    \item If $b>1$,
    \begin{equation*}
        \Theta_a(\tau)=\Theta_a(0)+O(\tau),
        \qquad
        \Theta_a'(\tau)=O(1).
    \end{equation*}
\end{itemize}
For every $m\geq1$ such that $a+m\xi\notin\mathbb Z$, the same conclusions
hold for $\mathcal S_\nu^m p_a$ after replacing $\xi$ by $m\xi$.
\end{proposition}

\begin{proof}
Changing variables $\lambda=\tau u$ in the defining integral gives
\begin{align*}
    \Theta_a(\tau)
    &=
    \frac{\beta^\xi}{\Gamma(\xi)}
    \int_0^\infty
        (\tau+\lambda)^a\lambda^{\xi-1}e^{-\beta\lambda}
    \,\diff\lambda
    \\
    &=
    \frac{\beta^\xi\tau^{a+\xi}}{\Gamma(\xi)}
    \int_0^\infty
        (1+u)^a u^{\xi-1}e^{-\beta\tau u}
    \,\diff u.
\end{align*}
The integral representation
\begin{equation*}
    U(c,d,z)
    =\frac{1}{\Gamma(c)}
    \int_0^\infty e^{-zu}u^{c-1}(1+u)^{d-c-1}\,\diff u
\end{equation*}
for $\operatorname{Re}(c)>0$ and $\operatorname{Re}(z)>0$
\cite[13.2.5]{HMF}, applied with $(c,d,z)=(\xi,b+1,\beta\tau)$, gives
\begin{equation*}
    \Theta_a(\tau)=\beta^\xi\tau^bU(\xi,b+1,\beta\tau).
\end{equation*}
The near-zero equivalents \cite[13.5.6 and 13.5.8]{HMF} for $b>0$ and
\cite[13.5.10 and 13.5.12]{HMF} for $b<0$ read
\begin{align*}
    U(\xi,b+1,z)
    &\sim\frac{\Gamma(b)}{\Gamma(\xi)}z^{-b},
    && b>0,
    \\
    U(\xi,b+1,z)
    &\longrightarrow\frac{\Gamma(-b)}{\Gamma(-a)},
    && b<0.
\end{align*}
Thus, when $b>0$,
\begin{equation*}
    \Theta_a(0)=\beta^{-a}\frac{\Gamma(b)}{\Gamma(\xi)},
\end{equation*}
whereas for $b<0$,
\begin{equation*}
    \Theta_a(\tau)
    =C_{a,\xi}\tau^b+o(\tau^b).
\end{equation*}
The coefficient $C_{a,\xi}$ is nonzero because the Gamma function has no
zeros. Differentiation under the integral gives
\begin{equation*}
    \Theta_a'=a\Theta_{a-1}.
\end{equation*}
If $b<1$, applying the second equivalent above with $a$ replaced by $a-1$
and using $aC_{a-1,\xi}=bC_{a,\xi}$ yields
\begin{equation*}
    \Theta_a'(\tau)
    =bC_{a,\xi}\tau^{b-1}+o(\tau^{b-1}).
\end{equation*}
For $0<b<1$, integration from zero gives the corresponding expansion of
$\Theta_a$. If $b>1$, then $\Theta_{a-1}(0)$ is finite, so
$\Theta_a'(\tau)=O(1)$ and $\Theta_a(\tau)=\Theta_a(0)+O(\tau)$. The integral
representation of $\mathcal S_\nu^m p_a$ is obtained by replacing $\xi$ with
$m\xi$, and the iterated statement is then a consequence of the previous
analysis.
\end{proof}

\subsection{Gamma smoothing of perturbations}

\begin{proposition}[Stability under Gamma smoothing]
\label{prop:gamma_smoothing_lower_order}
Let $a\in(-1/2,0)$ with $a+\xi\notin\mathbb Z$, and let
$f\in\mathcal C^1(\R_+^*)$ satisfy
\begin{equation*}
    f(\tau)=o(\tau^a),
    \qquad
    f'(\tau)=o(\tau^{a-1}),
    \qquad \tau\downarrow0.
\end{equation*}
Assume also that $f$ and $f'$ are bounded on $[r,\infty)$ for every $r>0$.
Then
\begin{equation*}
    \mathcal S_\nu f(\tau)=o(\tau^a),
    \qquad
    (\mathcal S_\nu f)'(\tau)
    =\mathcal S_\nu f'(\tau)
    =o(\tau^{a-1}).
\end{equation*}
Moreover,
\begin{itemize}
    \item Suppose $0<\xi<1$ and set $b:=a+\xi$. If $b<0$, then
    \begin{equation*}
        \mathcal S_\nu f(\tau)=o(\tau^b),
        \qquad
        (\mathcal S_\nu f)'(\tau)=o(\tau^{b-1}).
    \end{equation*}
    If $0<b<1$, then $\mathcal S_\nu f(0)$ is finite and
    \begin{equation*}
        \mathcal S_\nu f(\tau)
        -\mathcal S_\nu f(0)
        =o(\tau^b),
        \qquad
        (\mathcal S_\nu f)'(\tau)=o(\tau^{b-1}).
    \end{equation*}

    \item If $\xi\geq1$, then $\mathcal S_\nu f(0)$ is finite and
    \begin{equation*}
        \mathcal S_\nu f(\tau)
        -\mathcal S_\nu f(0)
        =o(\tau^{a+1}),
        \qquad
        (\mathcal S_\nu f)'(\tau)=o(\tau^a).
    \end{equation*}
\end{itemize}

For every $m\geq1$ such that $a+m\xi\notin\mathbb Z$, the same conclusions
hold for $\mathcal S_\nu^m f$ after replacing $\xi$ by $m\xi$.
\end{proposition}

\begin{proof}
Since $f'$ is bounded on every $[r,\infty)$, dominated convergence justifies
differentiation under the integral sign for every $\tau>0$. Hence,
\begin{equation*}
    (\mathcal S_\nu f)'=\mathcal S_\nu f'
    \qquad\text{on }\R_+^*.
\end{equation*}
It therefore remains to establish a common estimate for any
$h\in\mathcal C(\R_+^*)$, bounded on every $[r,\infty)$, such that
\begin{equation*}
    h(\tau)=o(\tau^\rho),
    \qquad
    \rho<0,
    \qquad
    \rho+\xi\notin\mathbb Z.
\end{equation*}
Given $\varepsilon>0$, choose $\delta>0$ such that
$|h(x)|\leq\varepsilon x^\rho$ for $0<x<2\delta$, and, for
$0<\tau<\delta$, split at $\lambda=\delta$:
\begin{equation*}
    \mathcal S_\nu h(\tau)
    =
    \underbrace{\int_0^\delta h(\tau+\lambda)\,\nu(\diff\lambda)}_{J_0(\tau)}
    +
    \underbrace{\int_\delta^\infty h(\tau+\lambda)\,\nu(\diff\lambda)}
        _{J_\infty(\tau)}.
\end{equation*}
Then
\begin{equation*}
    |J_0(\tau)|\leq\varepsilon\Theta_\rho(\tau),
\end{equation*}
whereas $J_\infty(\tau)=O(1)$ by the boundedness of $h$ on
$[\delta,\infty)$. Since $\rho+\xi\notin\mathbb Z$,
\cref{prop:gamma_asymp} and then $\varepsilon\downarrow0$ yield
\begin{equation*}
    \mathcal S_\nu h(\tau)
    =
    \begin{cases}
        o(\tau^{\rho+\xi}), & \rho+\xi<0,\\
        O(1), & \rho+\xi>0.
    \end{cases}
\end{equation*}
Since $\rho<0$, both scenarios imply
$\mathcal S_\nu h(\tau)=o(\tau^\rho)$. If $\rho+\xi>0$, then
$(\tau+\lambda)^\rho\leq\lambda^\rho$ and
$\lambda^{\rho+\xi-1}$ is integrable at zero. Dominated convergence
therefore gives the finite endpoint value and
$\mathcal S_\nu h(\tau)\to\mathcal S_\nu h(0)$.\\

We now apply the common estimate to $h=f$ with $\rho=a$ and to $h=f'$ with
$\rho=a-1$.
\begin{itemize}
    \item In every case, it gives
    $\mathcal S_\nu f(\tau)=o(\tau^a)$ and
    $\mathcal S_\nu f'(\tau)=o(\tau^{a-1})$.

    \item If $0<\xi<1$ and $b:=a+\xi$, then
    $(a-1)+\xi=b-1<0$ yields
    $(\mathcal S_\nu f)'(\tau)=o(\tau^{b-1})$. If $b<0$, then
    $\mathcal S_\nu f(\tau)=o(\tau^b)$, while if $b>0$,
    $\mathcal S_\nu f(0)<\infty$ and
    \begin{equation*}
        \mathcal S_\nu f(\tau)-\mathcal S_\nu f(0)
        =\int_0^\tau(\mathcal S_\nu f)'(u)\,\diff u
        =o(\tau^b).
    \end{equation*}

    \item If $\xi\geq1$, \cref{prop:gamma_asymp} gives
    $\Theta_{a-1}(\tau)=O(\tau^a)$, so that one has
    $(\mathcal S_\nu f)'(\tau)=o(\tau^a)$. Since $a+\xi>0$,
    $\mathcal S_\nu f(0)<\infty$, and, since $a>-1$,
    \begin{equation*}
        \mathcal S_\nu f(\tau)-\mathcal S_\nu f(0)
        =\int_0^\tau(\mathcal S_\nu f)'(u)\,\diff u
        =o(\tau^{a+1}).
    \end{equation*}
\end{itemize}
\end{proof}

\subsection{Criterion for Gaussian--Volterra regularity}

\iffalse
\begin{proposition}[Regularity of Gaussian--Volterra processes]
\label{prop:volterra_holder}
Let $b\in(-1/2,\infty)\setminus\{0\}$ and let
$g\in L^2_{\mathrm{loc}}(\R_+)\cap\mathcal C^1(\R_+^*)$. Assume that $g'$
is bounded on every compact subset of $\R_+^*$ and that, for some
$\delta>0$ and $C>0$, one of the following bound forms holds on
$(0,\delta]$:
\begin{itemize}
    \item if $b<0$,
    \begin{equation*}
        |g(\tau)|\leq C\tau^b,
        \qquad
        |g'(\tau)|\leq C\tau^{b-1};
    \end{equation*}
    \item if $b>0$, there is a constant $c_0\in\R$ such that
    \begin{equation*}
        |g(\tau)-c_0|\leq C\tau^b,
        \qquad
        |g'(\tau)|\leq C\tau^{b-1}.
    \end{equation*}
\end{itemize}
Set
\begin{equation*}
    \rho^*
    :=
    \begin{cases}
        b+\dfrac12, & b<0,\\[4pt]
        \dfrac12, & b>0\text{ and }c_0\neq0,\\[4pt]
        \min\!\left(1,b+\dfrac12\right),
            & b>0\text{ and }c_0=0.
    \end{cases}
\end{equation*}
Then the Volterra process
\begin{equation*}
    V_t:=\int_0^t g(t-s)\,\diff B_s
\end{equation*}
has a modification whose paths are locally $\rho$-H\"older for every
$\rho<\rho^*$.
\end{proposition}
\fi

\begin{proposition}[Regularity of Gaussian--Volterra processes]
\label{prop:volterra_holder}
Let $b\in(-1/2,\infty)\setminus\{0\}$ and let
$g\in L^2_{\mathrm{loc}}(\R_+)\cap\mathcal C^1(\R_+^*)$. Assume that $g'$
is bounded on every compact subset of $\R_+^*$ and that, for some
$\delta>0$ and $C>0$,
\begin{equation*}
    |g'(\tau)|\leq C\tau^{b-1},
    \qquad 0<\tau\leq\delta.
\end{equation*}
When $b>0$, this bound ensures that
$\lim_{\tau\downarrow0}g(\tau)$ exists.
Set
\begin{equation*}
    \rho^*
    :=
    \begin{cases}
        b+\dfrac12, & b<0,\\[4pt]
        \dfrac12,
            & b>0\text{ and }\lim_{\tau\downarrow0}g(\tau)\neq0,
            \\[4pt]
        \min\!\left(1,b+\dfrac12\right),
            & b>0\text{ and }\lim_{\tau\downarrow0}g(\tau)=0.
    \end{cases}
\end{equation*}
Then the Volterra process
\begin{equation*}
    V_t:=\int_0^t g(t-s)\,\diff B_s
\end{equation*}
has a modification whose paths are locally $\rho$-H\"older for every
$\rho<\rho^*$. Moreover, $V$ is almost surely not locally $\rho$-H\"older
for any $\rho>\rho^*$ provided one of the
following conditions holds:
\begin{enumerate}[label=(\roman*)]
    \item $b>0$ and $\lim_{\tau\downarrow0}g(\tau)\neq0$;
    \item there is a constant $C'>0$ such that
    \begin{equation}
    \label{eq:volterra_kernel_lower_bound}
        |g(\tau)|\geq C'\tau^b,
        \qquad 0<\tau\leq\delta.
\end{equation}
\end{enumerate}
\end{proposition}

\begin{proof}
When $b>0$, the derivative bound implies that
$c_0:=\lim_{u\downarrow0}g(u)$ exists and that
$|g(u)-c_0|\leq(C/b)u^b$, whereas when $b<0$, integration from $u$ to
$\delta$ gives $|g(u)|\leq C_0u^b$ on $(0,\delta]$ for some $C_0>0$.
Setting $\widetilde g:=g-c_0$ when $b>0$ and $\widetilde g:=g$ when $b<0$,
one has, after enlarging $C$,
\begin{equation*}
    |\widetilde g(\tau)|\leq C\tau^b,
    \qquad
    |\widetilde g'(\tau)|\leq C\tau^{b-1},
    \qquad 0<\tau\leq\delta.
\end{equation*}

For fixed $T>0$, choose $\eta\in(0,\min(\delta/2,T))$ and let
$0<h<\min(1,\eta)$. It\^o's isometry gives
\begin{align*}
    \E\bigl[|V_{t+h}-V_t|^2\bigr]
    ={}&
    \underbrace{\int_0^h g(u)^2\,\diff u}_{I_1(h)}
    \\
    &+
    \underbrace{\int_0^t|g(u+h)-g(u)|^2\,\diff u}_{I_2(t,h)}.
\end{align*}
Since $t\leq T$ and the integrand is nonnegative, split
\begin{equation*}
    \begin{aligned}
        I_2(t,h)
        \leq{}&
        \underbrace{\int_0^h
            |\widetilde g(u+h)-\widetilde g(u)|^2\,\diff u}_{J_0(h)}
        +
        \underbrace{\int_h^\eta
            |\widetilde g(u+h)-\widetilde g(u)|^2\,\diff u}_{J_1(h)}
        \\
        &+
        \underbrace{\int_\eta^T
            |\widetilde g(u+h)-\widetilde g(u)|^2\,\diff u}_{J_\infty(h)}.
    \end{aligned}
\end{equation*}
The function bound gives
\begin{equation*}
    J_0(h)=O\!\left(h^{2b+1}\right).
\end{equation*}
Since $u+h\leq2u$ on $[h,\eta]$, the mean value theorem gives
\begin{equation*}
    J_1(h)
    =O\!\left(h^2\int_h^\eta u^{2b-2}\,\diff u\right)
    =
    \begin{cases}
        O\!\left(h^{2b+1}\right), & b<1/2,\\
        O\!\left(h^2\log(1/h)\right), & b=1/2,\\
        O\!\left(h^2\right), & b>1/2.
    \end{cases}
\end{equation*}
Finally, set $M_T:=\sup_{[\eta,T+1]}|g'|<\infty$. The mean value theorem gives
$J_\infty(h)\leq TM_T^2h^2$. There then exists a constant $C_T$ such that
the three bounds combine into
\begin{equation*}
    I_2(t,h)
    \leq
    \begin{cases}
        C_T h^{2b+1}, & b<1/2,\\
        C_T h^2\log(1/h), & b=1/2,\\
        C_T h^2, & b>1/2.
    \end{cases}
\end{equation*}

The first term satisfies
\begin{equation*}
    I_1(h)
    \leq
    \begin{cases}
        C h^{2b+1}, & b<0,\\
        C h, & b>0\text{ and }c_0\neq0,\\
        C h^{2b+1}, & b>0\text{ and }c_0=0.
    \end{cases}
\end{equation*}
Combining these estimates shows that, for every
$\bar\rho<\rho^*$, there is a constant $C_{T,\bar\rho}$ such that
\begin{equation*}
    \E\bigl[|V_{t+h}-V_t|^2\bigr]
    \leq C_{T,\bar\rho}h^{2\bar\rho},
    \qquad t\in[0,T].
\end{equation*}
At $b=1/2$, use
$h^2\log(1/h)=O(h^{2\bar\rho})$ for every $\bar\rho<1$.
Since $V$ is centred Gaussian, the preceding increment bound and
Kolmogorov's continuity criterion yield a modification whose paths are
locally $\rho$-H\"older for every $\rho<\rho^*$.\\

It remains to prove that $V$ is almost surely not locally $\rho$-H\"older for
any $\rho>\rho^*$ under (i) or (ii). If
$b>0$ and $c_0\neq0$, then
\begin{equation*}
    \E[|V_h|^2]
    =\int_0^h g(u)^2\,\diff u
    \sim c_0^2h.
\end{equation*}
Otherwise, if \eqref{eq:volterra_kernel_lower_bound} holds and $b<1/2$, then
\begin{equation*}
    \E[|V_h|^2]
    \geq \frac{(C')^2}{2b+1}h^{2b+1}.
\end{equation*}
In this case, $\rho^*=b+1/2$.
Thus, in either case, for all sufficiently small $h$,
\begin{equation}
\label{eq:volterra_variance_lower_sharpness}
    \E[|V_h-V_0|^2]\geq ch^{2\rho^*},
    \qquad \rho^*<1.
\end{equation}
Together with \eqref{eq:volterra_variance_lower_sharpness}, the necessary
condition for Gaussian H\"older continuity in
\cite[Theorem~1]{AzmoodehSottinenViitasaariYazigi2014} implies that, almost
surely, $V$ is not locally $\rho$-H\"older for any $\rho>\rho^*$.\\

Finally, suppose that \eqref{eq:volterra_kernel_lower_bound} holds, $c_0=0$,
and $b\geq1/2$. Then
$\rho^*=1$, and the preceding argument gives local $\rho$-H\"older continuity
for every $\rho<1$. Since a function that is H\"older of order greater than one
on an interval is constant, the almost-sure nonconstancy of $V$ rules out every
$\rho>1$.
\end{proof}

\section{Proofs of Section~\ref{sec:scaling_limit}}
\label{app:scaling_proofs}

\begin{lemma}[Decay of Abel averages]
\label{lem:scaling_abel_tail}
For $0\le\alpha<1$, define
\begin{equation*}
\mathcal A_\alpha h
:=
(1-\alpha)\sum_{n\ge1}\alpha^{n-1}\mathcal S_\nu^n h.
\end{equation*}
For every $p>0$, the operator
$\mathcal A_\alpha:\mathcal B_p\to\mathcal B_0$ satisfies
\begin{equation}
\label{eq:abel_tail_bound}
\norm{\mathcal A_\alpha}_{\mathcal B_p\to\mathcal B_0}
\le
\eta_p(\alpha),
\end{equation}
where
\begin{equation*}
\eta_p(\alpha)
:=
\E\!\left[
\left(1+\Lambda_1+\cdots+\Lambda_{N_\alpha}\right)^{-p}
\right],
\qquad
\Prob(N_\alpha=n)=(1-\alpha)\alpha^{n-1},\qquad n\ge1,
\end{equation*}
and the variables $\Lambda_i$ are independent with law $\nu$ and independent
of $N_\alpha$. Moreover,
$\eta_p(\alpha)\to0$ as $\alpha\uparrow1$.
\end{lemma}

\begin{proof}
Fix $h\in\mathcal B_p$ and let
$S_n:=\Lambda_1+\cdots+\Lambda_n$. Since
\begin{equation*}
\mathcal S_\nu^n h(t)=\E[h(t+S_n)],
\end{equation*}
and $\mathcal S_\nu$ is a contraction on $\mathcal B_0$, the Abel series
converges in $\mathcal B_0$ for every fixed $\alpha<1$. Conditioning on
$N_\alpha$ gives
\begin{equation*}
\mathcal A_\alpha h(t)=\E[h(t+S_{N_\alpha})].
\end{equation*}
Consequently,
\begin{align*}
|\mathcal A_\alpha h(t)|
&\le
\norm{h}_p
\E\!\left[(1+t+S_{N_\alpha})^{-p}\right]
\\
&\le
\eta_p(\alpha)\norm{h}_p,
\end{align*}
which proves \eqref{eq:abel_tail_bound}. Also,
$N_\alpha\to\infty$ in probability as $\alpha\uparrow1$. Since
$\nu(\mathbb R_+^*)=1$, one has $S_n\to\infty$ almost surely, and hence
$S_{N_\alpha}\to\infty$ in probability. The bounded random variables
$(1+S_{N_\alpha})^{-p}$ therefore converge to zero in $L^1$, proving
$\eta_p(\alpha)\to0$.

\end{proof}

\begin{proposition}[Bounds for the impact operator]
\label{prop:scaling_impact_bounds}
Let $p>0$ and suppose
\begin{equation*}
    \mathcal S_\nu^{n_0}K\in\mathcal B_p,
    \qquad
    \int_0^\infty(1+t)^p|G'(t)|\,\diff t<\infty,
    \qquad
    \norm{G'}_{L^1(\mathbb R_+)}<1,
\end{equation*}
and
\begin{equation*}
    \mathcal H(\mathcal S_\nu^jK)\in\mathcal B_p,
    \qquad 0\le j<n_0.
\end{equation*}
Then $\mathcal H$ is bounded from $\mathcal B_p$ to
$\mathcal B_p$, from $\mathcal Y_K^p$ to $\mathcal B_p$, and from
$\mathcal Y_K^0$ to $\mathcal B_0$. Moreover,
$\mathcal H:\mathcal B_0\to\mathcal B_0$ is a strict
contraction, with
\begin{equation*}
    \norm{\mathcal H}_{\mathcal B_0\to\mathcal B_0}
    \leq\norm{G'}_{L^1(\mathbb R_+)}<1.
\end{equation*}
\end{proposition}

\begin{proof}
For $b\in\mathcal B_p$, the inequality
$1+t\le(1+s)(1+t-s)$, valid for $0\le s\le t$, gives
\begin{equation*}
\norm{\mathcal Hb}_p
\le
\left(
\int_0^\infty(1+s)^p|G'(s)|\,\diff s
\right)\norm{b}_p.
\end{equation*}
On $\mathcal B_0$, convolution with the $L^1$-kernel $G'$
preserves vanishing at infinity and satisfies
\begin{equation*}
\norm{\mathcal Hb}_{\mathcal B_0}
\le
\norm{G'}_{L^1(\mathbb R_+)}
\norm{b}_{\mathcal B_0}.
\end{equation*}
This proves the assertions on $\mathcal B_p$ and $\mathcal B_0$. The
direct-sum decomposition and the assumed membership of
$\mathcal H(\mathcal S_\nu^jK)$ in $\mathcal B_p$ give
\begin{equation*}
\norm{\mathcal Hf}_p
\leq C\norm{f}_{\mathcal Y_K^p}.
\end{equation*}
Thus $\mathcal H:\mathcal Y_K^p\to\mathcal B_p$ is bounded.
The same finite-dimensional argument, using
$\mathcal B_p\subset\mathcal B_0$, proves boundedness from
$\mathcal Y_K^0$ to $\mathcal B_0$.
\end{proof}

\begin{proof}[Proof of \Cref{prop:scaling_leading}]
Set $q:=\norm{G'}_{L^1(\mathbb R_+)}<1$.
By \cref{prop:scaling_impact_bounds}, $\mathcal H$ maps
$\mathcal Y_K^0$ boundedly into $\mathcal B_0$ and has
$\mathcal B_0$-operator norm at most $q$. Since
$\mathcal B_0\hookrightarrow\mathcal Y_K^0$ is isometric, for every $m\geq1$,
\begin{equation*}
\norm{\mathcal H^{\,m}}_{\mathcal Y_K^0\to\mathcal Y_K^0}
\leq
\norm{\mathcal H}_{\mathcal Y_K^0\to\mathcal B_0}
q^{m-1}.
\end{equation*}
The Neumann series
\begin{equation*}
R_\alpha
:=
\mathrm{Id}+\sum_{m\ge1}(-\alpha)^m\mathcal H^{\,m}
\end{equation*}
therefore converges uniformly in operator norm for $\alpha\in[0,1]$ and
defines $R_\alpha=(\mathrm{Id}+\alpha\mathcal H)^{-1}$ on
$\mathcal Y_K^0$, with
\begin{equation*}
\sup_{\alpha\in[0,1]}
\norm{R_\alpha}_{\mathcal Y_K^0\to\mathcal Y_K^0}
\leq
1+
\frac{
\norm{\mathcal H}_{\mathcal Y_K^0\to\mathcal B_0}
}{1-q}.
\end{equation*}

Put $E_\alpha K:=Q_\alpha K+K$. Since
$\mathcal H\Phi_\alpha\in\mathcal B_p\subset\mathcal B_0$,
substituting the Abel identity
$Q_\alpha h=-h+\mathcal A_\alpha h$ into
\eqref{eq:v2_scaling_Q_resolvent} and applying $R_\alpha$ gives
\begin{equation*}
\Phi_\alpha
=
R_\alpha
\left[
-K+E_\alpha K
+\alpha\mathcal A_\alpha\mathcal H\Phi_\alpha
\right].
\end{equation*}
We now pass to the limit in this representation. The resolvent identity
\begin{equation*}
R_\alpha-R_1
=(1-\alpha)R_\alpha\mathcal HR_1
\end{equation*}
and the preceding uniform bound show that $R_\alpha\to R_1$ in operator norm.
Since $K$ belongs to $\mathcal Y_K^0$,
\cref{lem:app_Qc_extension} gives
\begin{equation*}
E_\alpha K\longrightarrow0
\qquad\text{in }\mathcal Y_K^0.
\end{equation*}
Moreover, \eqref{eq:v2_weighted_phi_bound} and
\cref{prop:scaling_impact_bounds} give
\begin{equation*}
\sup_{\alpha\in[\alpha_0,1)}
\norm{\mathcal H\Phi_\alpha}_p<\infty.
\end{equation*}
Hence \cref{lem:scaling_abel_tail} yields
\begin{equation*}
\mathcal A_\alpha\mathcal H\Phi_\alpha
\longrightarrow0
\qquad\text{in }\mathcal B_0.
\end{equation*}
The convergence also holds in $\mathcal Y_K^0$ by the isometric embedding
$\mathcal B_0\hookrightarrow\mathcal Y_K^0$. Therefore,
\begin{equation*}
\Phi_\alpha
\longrightarrow
-R_1K
=
-(\mathrm{Id}+\mathcal H)^{-1}K
\end{equation*}
in $\mathcal Y_K^0$. This proves \eqref{eq:Phi_alpha_limit} and
\eqref{eq:v2_leading_volterra}, and uniqueness follows from the
operator-norm Neumann construction of $R_1$.
Finally, continuity of $\mathrm{Id}+\mathcal H$ and
\eqref{eq:v2_observed_kernel_cancellation_curve} give
\begin{equation*}
\mathcal K_\alpha
=K+(\mathrm{Id}+\mathcal H)\Phi_\alpha
\longrightarrow
K+(\mathrm{Id}+\mathcal H)\Phi
=0
\qquad\text{in }\mathcal Y_K^0,
\end{equation*}
which proves \eqref{eq:K_alpha_unscaled_limit}.
\end{proof}

\begin{proof}[Proof of \Cref{prop:scaling_ou_uniform_tails}]
Choose $0<\eta<\min\{\gamma,\kappa(1-\Theta)\}$ and set
\begin{align*}
\mathcal E_\eta
&:=
\left\{
f\in\mathcal B_0:
\norm{f}_{\mathcal E_\eta}:=\sup_{t>0}e^{\eta t}|f(t)|<\infty
\right\},
\\
\mathcal L_\nu(x)
&:=\int_{(0,\infty)}e^{-x\lambda}\,\nu(\diff\lambda),
\qquad x>0.
\end{align*}
Since $\mathcal L_\nu(\eta)<1$ and
$\norm{\mathcal S_\nu f}_{\mathcal E_\eta}\leq
\mathcal L_\nu(\eta)\norm{f}_{\mathcal E_\eta}$, the definition of
$\mathcal A_\alpha$
gives
\begin{equation*}
\norm{\mathcal A_\alpha}_{\mathcal E_\eta\to\mathcal E_\eta}
\leq
(1-\alpha)\sum_{n\geq1}\alpha^{n-1}\mathcal L_\nu(\eta)^n
=\frac{(1-\alpha)\mathcal L_\nu(\eta)}
{1-\alpha\mathcal L_\nu(\eta)}
\longrightarrow0.
\end{equation*}
It follows that
\begin{equation*}
\norm{Q_\alpha}_{\mathcal E_\eta\to\mathcal E_\eta}
\leq
1+\frac{(1-\alpha)\mathcal L_\nu(\eta)}
{1-\alpha\mathcal L_\nu(\eta)},
\qquad
\norm{\mathcal H}_{\mathcal E_\eta\to\mathcal E_\eta}
\leq\frac{\Theta\kappa}{\kappa-\eta}<1.
\end{equation*}
Thus, we may choose $\alpha_0<1$ and $r<1$ such that
\begin{equation*}
\sup_{\alpha\in[\alpha_0,1)}
\norm{\alpha Q_\alpha\mathcal H}_{\mathcal E_\eta\to\mathcal E_\eta}
\leq r.
\end{equation*}

As $\mathcal S_\nu K=\mathcal L_\nu(\gamma)K$, one has
\begin{equation*}
Q_\alpha K
=-
\frac{1-\mathcal L_\nu(\gamma)}{1-\alpha\mathcal L_\nu(\gamma)}K,
\qquad
\norm{Q_\alpha K}_{\mathcal E_\eta}\leq1.
\end{equation*}
Hence \eqref{eq:v2_scaling_Q_resolvent} gives
\begin{align*}
\Phi_\alpha
&=\sum_{m\geq0}
(\alpha Q_\alpha\mathcal H)^mQ_\alpha K,
\\
\norm{\Phi_\alpha}_{\mathcal E_\eta}
&\leq\sum_{m\geq0}r^m
=\frac{1}{1-r},
\qquad \alpha\in[\alpha_0,1).
\end{align*}
By \cref{rem:contraction_permanent_impact}, the assumption
$\Theta\leq1/2$ identifies these solutions with the selected equilibrium
kernels. Finally, $n_0=0$ and
$\norm{f}_p\leq\sup_{t>0}(1+t)^pe^{-\eta t}
\norm{f}_{\mathcal E_\eta}$, so the exponential bound implies
\eqref{eq:v2_weighted_phi_bound} for every $p>0$.
\end{proof}

\section{Proofs of Section~\ref{sec:tail_regularity}}
\label{app:regularity_proofs}

\subsection{Kernel asymptotics}

\begin{proof}[Proof of \Cref{prop:v2_phi_asymp}]
Let $\phi_\alpha$ be Gamma-regular. For every
$1\leq m<n_0$,
\cref{prop:gamma_asymp,prop:gamma_smoothing_lower_order} give
\begin{equation*}
    \mathcal S_\nu^mK(\tau)=o(\tau^a),
    \qquad
    (\mathcal S_\nu^mK)'(\tau)=o(\tau^{a-1}).
\end{equation*}
Combining these estimates,
\eqref{eq:v2_rough_kernel_specification}, and the defining estimates on
$r_\alpha$ in \eqref{eq:v2_phi_expansion} gives
\begin{align*}
    \phi_\alpha(\tau)
    &=
    -\frac{\tau^a}{1+\alpha G(0)}+o(\tau^a),
    \\
    \phi_\alpha'(\tau)
    &=
    -\frac{a\tau^{a-1}}{1+\alpha G(0)}
    +o(\tau^{a-1}),
\end{align*}
which proves both kernel claims. The nonzero leading coefficient gives
$|\phi_\alpha(\tau)|\geq C\tau^a$ for all sufficiently small $\tau$.
Applying \cref{prop:volterra_holder} with $b=a<0$ gives the asserted local
regularity of $\Pi^\alpha$.
\end{proof}

The proof of \cref{prop:v2_phi_asymp} and the global assumptions imply that
$\phi_\alpha,\phi_\alpha'$ are bounded on $[r,\infty)$ for every $r>0$.
Near the origin,
\begin{align}
\label{eq:app_d_Hphi_expansion}
    \mathcal H\phi_\alpha(\tau)
    &=\gamma_\alpha\tau^{a+1}+o(\tau^{a+1}),
    \\
    (\mathcal H\phi_\alpha)'(\tau)
    &=(a+1)\gamma_\alpha\tau^a+o(\tau^a),
\end{align}
where
\begin{equation*}
    \gamma_\alpha
    :=-\frac{G'(0)}{(1+\alpha G(0))(a+1)}.
\end{equation*}
These relations follow from the expansion of $\phi_\alpha$, the identity
$\mathcal H\phi_\alpha=G'\ast\phi_\alpha$, and
\begin{equation*}
    (\mathcal H\phi_\alpha)'(\tau)
    =G'(0)\phi_\alpha(\tau)
     +(G''\ast\phi_\alpha)(\tau).
\end{equation*}
In addition, for every $r>0$, the local integrability of $\phi_\alpha$, its
boundedness on $[r,\infty)$, and the standing assumptions imply that
$\mathcal H\phi_\alpha$ and $(\mathcal H\phi_\alpha)'$ are bounded on
$[r,\infty)$, as required in
\cref{prop:gamma_smoothing_lower_order}.

\iffalse
\begin{proof}
The Gamma-regular representation and the standing assumptions imply that
$\phi_\alpha$ and $\phi_\alpha'$ are bounded on $[r,\infty)$ for every
$r>0$. Since $\phi_\alpha\in L^1_{\mathrm{loc}}$,
\begin{equation*}
    (\mathcal H\phi_\alpha)'(\tau)
    =G'(0)\phi_\alpha(\tau)
     +(G''\ast\phi_\alpha)(\tau).
\end{equation*}
By \cref{prop:v2_phi_asymp} and the boundedness of $G''$,
$(G''\ast\phi_\alpha)(\tau)=O(\tau^{a+1})=o(\tau^a)$. This gives the
derivative expansion, and integration from zero gives the function
expansion. The same identity, the local integrability of $\phi_\alpha$, and
the global assumptions on $G'$ and $G''$ give the asserted bounds on
$[r,\infty)$.

The derivative is therefore bounded by $C\tau^a$ near zero and is bounded on
every $[r,\infty)$. Differentiation under the Gamma integral and
\cref{prop:gamma_asymp} give
\begin{equation*}
    \bigl|(\mathcal S_\nu\mathcal H\phi_\alpha)'(\tau)\bigr|
    \leq C\Theta_a(\tau)+O(1)
    =o(\tau^a).
\end{equation*}
The local expansion and global boundedness of
$\mathcal H\phi_\alpha$ make $c_H$ finite. Integrating the derivative
estimate gives the stated expansion.
\end{proof}
\fi

\begin{proof}[Proof of \Cref{prop:v2_regularity_alpha}]
Let
\begin{equation*}
    A_\alpha
    :=\frac{1-(1-\alpha)G(0)}{1+\alpha G(0)}\neq0.
\end{equation*}
The identity \eqref{eq:amplitude_identity} reads
\begin{equation*}
    \mathcal K_\alpha=A_\alpha K+\mathcal R_\alpha.
\end{equation*}
We first show that
\begin{equation*}
    \mathcal R_\alpha(\tau)=o(\tau^a),
    \qquad
    \mathcal R_\alpha'(\tau)=o(\tau^{a-1}).
\end{equation*}
Set
\begin{equation*}
    q_\alpha
    :=\phi_\alpha+\frac{\tau^a}{1+\alpha G(0)}.
\end{equation*}
By \cref{prop:v2_phi_asymp} and the Gamma-regular bounds,
$q_\alpha,q_\alpha'$ are bounded on every $[r,\infty)$ and
\begin{equation*}
    q_\alpha=o(\tau^a),
    \qquad
    q_\alpha'=o(\tau^{a-1}).
\end{equation*}
Applying \cref{prop:gamma_asymp,prop:gamma_smoothing_lower_order} to
$K=\tau^a+K_s$ and
$\phi_\alpha=-\tau^a/(1+\alpha G(0))+q_\alpha$ gives
\begin{align*}
    \mathcal S_\nu K&=o(\tau^a),
    &
    (\mathcal S_\nu K)'&=o(\tau^{a-1}),
    \\
    \mathcal S_\nu\phi_\alpha&=o(\tau^a),
    &
    (\mathcal S_\nu\phi_\alpha)'&=o(\tau^{a-1}).
\end{align*}
Equation \eqref{eq:app_d_Hphi_expansion} gives the sharper bounds
$\mathcal H\phi_\alpha(\tau)=O(\tau^{a+1})$ and
$(\mathcal H\phi_\alpha)'(\tau)=O(\tau^a)$. Moreover,
$\mathcal D_\nu\mathcal H\phi_\alpha
 =\mathcal S_\nu\mathcal H\phi_\alpha-\mathcal H\phi_\alpha$ is bounded near
the origin, while its derivative is $o(\tau^{a-1})$ by
\cref{prop:gamma_asymp,prop:gamma_smoothing_lower_order} applied to
$(\mathcal H\phi_\alpha)'=O(\tau^a)$.
Substitution into \eqref{eq:amplitude_remainder} proves the required lower
order of $\mathcal R_\alpha$.
Using \eqref{eq:v2_rough_kernel_specification}, we obtain
\begin{align*}
    \mathcal K_\alpha(\tau)
    &=A_\alpha\tau^a+o(\tau^a),
    \\
    \mathcal K_\alpha'(\tau)
    &=A_\alpha a\tau^{a-1}+o(\tau^{a-1}).
\end{align*}
Since $A_\alpha\neq0$, the first equivalent gives
$|\mathcal K_\alpha(\tau)|\geq C\tau^a$ for all sufficiently small $\tau$.
Applying \cref{prop:volterra_holder} shows that $X^\alpha$ has local H\"older regularity
$a+1/2=H$.
\end{proof}

\begin{proof}[Proof of \Cref{prop:v2_bottleneck}]
Under the cancellation condition, write
\begin{equation*}
    \phi_\alpha
    =-\frac{K}{1+\alpha G(0)}+\rho_\alpha.
\end{equation*}
The Gamma-regular expansion gives
\begin{equation*}
    \rho_\alpha=o(\tau^a),
    \qquad
    \rho_\alpha'=o(\tau^{a-1}),
\end{equation*}
and $\rho_\alpha,\rho_\alpha'$ are bounded on every $[r,\infty)$, so that
\cref{prop:gamma_smoothing_lower_order} applies.
Substitution into \eqref{eq:amplitude_remainder} gives
\begin{equation*}
    \mathcal K_\alpha
    =
    (1-\alpha)\mathcal S_\nu K
    +\alpha G(0)\mathcal S_\nu\rho_\alpha
    +V_\alpha,
\end{equation*}
where
\begin{equation*}
    V_\alpha
    :=(1-\alpha)\mathcal H\phi_\alpha
      +\alpha\mathcal S_\nu\mathcal H\phi_\alpha.
\end{equation*}
The expansion \eqref{eq:app_d_Hphi_expansion} and the global bounds show that
$c_H:=(\mathcal S_\nu\mathcal H\phi_\alpha)(0)$ is finite. Applying
\cref{prop:gamma_asymp,prop:gamma_smoothing_lower_order} to the derivative
expansion in \eqref{eq:app_d_Hphi_expansion} and integrating gives
\begin{equation*}
    \mathcal S_\nu\mathcal H\phi_\alpha(\tau)
    =c_H+o(\tau^{a+1}),
    \qquad
    (\mathcal S_\nu\mathcal H\phi_\alpha)'(\tau)=o(\tau^a).
\end{equation*}
Consequently,
\begin{align*}
    V_\alpha(\tau)
    &=\alpha c_H+(1-\alpha)\gamma_\alpha\tau^{a+1}
      +o(\tau^{a+1}),
    \\
    V_\alpha'(\tau)
    &=(1-\alpha)(a+1)\gamma_\alpha\tau^a+o(\tau^a).
\end{align*}

We distinguish two horizon regimes.

\begin{itemize}
\item \emph{Case $0<\xi<1$.}
Set
\begin{equation*}
    b:=a+\xi\neq0,
    \qquad
    \kappa_{a,\xi}
    :=\beta^\xi\frac{\Gamma(-a-\xi)}{\Gamma(-a)}
    \neq0.
\end{equation*}
By \cref{prop:gamma_asymp,prop:gamma_smoothing_lower_order},
\begin{equation*}
    \mathcal S_\nu K(\tau)
    =
    \begin{cases}
        \kappa_{a,\xi}\tau^b+o(\tau^b), & b<0,\\
        c_K+\kappa_{a,\xi}\tau^b+o(\tau^b), & b>0,
    \end{cases}
\end{equation*}
and
\begin{equation*}
    (\mathcal S_\nu K)'(\tau)
    =b\kappa_{a,\xi}\tau^{b-1}+o(\tau^{b-1}).
\end{equation*}
\Cref{prop:gamma_smoothing_lower_order} also gives
\begin{equation*}
    \mathcal S_\nu\rho_\alpha
    =
    \begin{cases}
        o(\tau^b), & b<0,\\
        c_\rho+o(\tau^b), & b>0,
    \end{cases}
\end{equation*}
with
\begin{equation*}
    (\mathcal S_\nu\rho_\alpha)'(\tau)=o(\tau^{b-1}).
\end{equation*}
Since $b<a+1$ and $b-1<a$, the nonconstant part of $V_\alpha$ and its
derivative are lower order. When $b<0$, its constant part is also
$o(\tau^b)$. Hence
\begin{equation*}
    \mathcal K_\alpha(\tau)
    =
    \begin{cases}
        C_\alpha\tau^b+o(\tau^b), & b<0,\\
        \ell_\alpha+C_\alpha\tau^b+o(\tau^b), & b>0,
    \end{cases}
\end{equation*}
with
\begin{equation*}
    C_\alpha=(1-\alpha)\kappa_{a,\xi}\neq0,
\end{equation*}
and the differentiated estimates give
\begin{equation*}
    \mathcal K_\alpha'(\tau)
    =bC_\alpha\tau^{b-1}+o(\tau^{b-1}).
\end{equation*}
When $b<0$, or when $b>0$ and $\ell_\alpha=0$, the nonzero coefficient gives
$|\mathcal K_\alpha(\tau)|\geq C\tau^b$ for all sufficiently small $\tau$.
When $b>0$ and $\ell_\alpha\neq0$, the kernel has a nonzero limit at the origin.
The conclusion follows from \cref{prop:volterra_holder}.

\item \emph{Case $\xi\geq1$.}
\Cref{prop:gamma_asymp,prop:gamma_smoothing_lower_order} give
\begin{align*}
    \mathcal S_\nu K(\tau)
    &=c_K+O(\tau^{a+1}),
    &
    (\mathcal S_\nu K)'(\tau)
    &=O(\tau^a),
    \\
    \mathcal S_\nu\rho_\alpha(\tau)
    &=c_\rho+o(\tau^{a+1}),
    &
    (\mathcal S_\nu\rho_\alpha)'(\tau)
    &=o(\tau^a).
\end{align*}
Together with the estimates for $V_\alpha$, this gives
\begin{equation*}
    \mathcal K_\alpha(\tau)=\ell_\alpha+O(\tau^{a+1}),
    \qquad
    \mathcal K_\alpha'(\tau)=O(\tau^a).
\end{equation*}
\Cref{prop:volterra_holder}, applied with $b=a+1$, gives local
H\"older regularity $1/2$ when $\ell_\alpha\neq0$. When
$\ell_\alpha=0$, it yields local $\rho$-H\"older continuity for every
$\rho<1$.
\end{itemize}
\end{proof}

\end{document}